\documentclass{article}
\usepackage[margin=1in]{geometry}
\usepackage{subfiles}
\usepackage[nobysame]{amsrefs}
\usepackage{amssymb}
\usepackage{amsthm}
\usepackage{amsmath}
\usepackage{amsfonts}
\usepackage[dvipsnames]{xcolor}
\usepackage{enumerate}
\usepackage{graphicx}
\usepackage{multicol}
\usepackage{verbatim}
\usepackage{hyperref}
\usepackage{amstext}
\usepackage{comment}
\usepackage{algpseudocode}
\usepackage{algorithm}
\usepackage{tikz}
\usetikzlibrary{arrows}
\usetikzlibrary{positioning}
\tikzset{main node/.style={circle,fill=white,draw,minimum size=0.5cm,inner sep=0pt},
}

\newtheorem{theorem}{Theorem}[section] 
\newtheorem{proposition}[theorem]{Proposition}
\newtheorem{lemma}[theorem]{Lemma}

\theoremstyle{definition}
\newtheorem{definition}[theorem]{Definition}

\theoremstyle{remark}
\newtheorem{remark}[theorem]{Remark}
\newtheorem{example}[theorem]{Example}

\title{Computation of Minimum Numbers of Tile and Bond-Edge Types for DNA Self-Assembly of Select Archimedean Graphs}
\author{Tabitha Merrithew, Jessica Sorrells \\ \small Converse University}

\begin{document}

\maketitle

\begin{abstract} This project mathematically models the self-assembly of DNA nanostructures in the shape of select Archimedean graphs using the flexible tile model. Under three different sets of restrictions called scenarios, we employ principles of linear algebra and graph theory to determine the minimum number of different DNA branched molecules and bond types needed to construct the desired shapes, theoretically reducing laboratory costs and the waste of biomaterials. We determine exact values for $T_3(G)$, the minimum number of molecule (or ``tile") types needed for all six order 12 and 24 Archimedean graphs. We also determine exact values for $B_3(G)$, the minimum number of strand (or ``bond-edge") types, for three of the six graphs and establish bounds for the remaining three. Two algorithms, implemented as Python scripts, are used to analyze proposed design strategies for the graphs. 
\end{abstract}

\section{Introduction}  

DNA nanostructures were first introduced in 1991 in the laboratory of Nadrian Seeman, where branched DNA molecules with sticky ends were ``ligated together" to construct a DNA nanocube \cite{nanocube}. Such nanostructures are now used in a variety of niche biological fields. For example, multidimensional DNA probes are used as biosensors and in bioimaging \cite{probes}. Other DNA nanostructures have been used to boost the immune system's ability to target and degrade cancer-related factors in cells \cite{cancer}. Tetrahedral DNA nanostructures, in particular, have versatile applications, including disease diagnosis through microRNA detection and aiding in the movement of stem cells to promote tissue regeneration \cite{microRNA, stemcell}. Given the scale, it is difficult and costly to manually direct the construction of these DNA objects. As a result, laboratories have turned to a self-assembly process paired with strategic design of component molecules. Self-assembled DNA structures resembling octahedra \cite{rothemund2006folding, zhang1994construction}, tetrahedra, and dodecahedra \cite{he2008hierarchical} have been achieved with these methods. Indeed, the structures of interest to laboratories are often skeletons of polyhedra, meshes, or other formations that can naturally be represented as discrete graphs. 

In the early 2000s, Jonoska et al. developed a graph theoretical formalism for describing DNA self-assembly, thus allowing for mathematical design of the idealized molecular building blocks (called ``tiles") capable of self-assembling into various geometric structures \cite{jonoska, jonoska2006spectrum}. These molecules are assumed to have flexible sticky ends that can bond to any other sticky end with a complementary sequence of base pairs (forming a ``bond-edge"). Ellis Monaghan et al. expanded upon this ``flexible tile model" and contextualized the framework into three different sets of constraints called ``scenarios" \cite{mintiles, ellis2019tile}.  Establishing optimal design strategies within each scenario can theoretically reduce waste of biomaterials. Discussions of computational complexity when computing within this graph theoretical model can be found in \cite{almodovar2021complexity, jonoska, JM09}. A broader overview of the collection of problems that have been generated by DNA self-assembly, referred to as ``DNA Mathematics," is given in \cite{Ellis_Next40}.

Here, we work within the flexible tile model, finding optimal designs for the skeletons of the six smallest Archimedean solids under the most strict set of conditions (``Scenario 3"). Hence, we seek to determine values of the parameters called $B_3$ and $T_3$, the minimum numbers of bond-edge and tile types in this scenario. We reproduce details of the flexible tile model and these parameters in the remainder of this section. This study of Archimedean graphs contributes to the knowledge of $k$-regular graphs in the flexible tile model. The polyhedral structure of these graphs makes them likely candidates for biological applications, such as drug delivery, as they can serve as a nanocage for pharmaceutical molecule cargo \cite{applications}. In particular, the truncated octahedron has been shown to have unique cell-targeting abilities, as it is geometrically similar to some molecules that bond to cell receptors \cite{oxLDL}. In Section \ref{sec:methods} we give an overview of methods on which we rely to devise design strategies, including two Python scripts and a characteristic, referred to as ``unswappable," of some graphs within this model. In Sections \ref{unswap} and \ref{swap} we compute exact values for $T_3$ for all six graphs, exact values for $B_3$ for three of the graphs, and ranges for $B_3$ for the remaining three graphs. Our findings are summarized in Table 1. 
\begin{table}[hbt!]
    \centering
    \begin{tabular}{| l || c | c | c |}
    \hline
    \multicolumn{4}{|c|}{Unswappable} \\ \hline 
     & Cuboctahedron & Sm. Rhombicuboctahedron & Snub Cube \\ \hline 
    $B_3(G)$ & 8 & $16 \leq B_3(G) \leq 17$ & $16$  \\ \hline
    $T_3(G)$ & 12 & $24 = \#V(G)$ & $24 = \#V(G)$ \\ \hline \hline
     \multicolumn{4}{|c|}{Swappable} \\ \hline 
     & Trunc. Tetrahedron & Trunc. Cube & Trunc. Octahedron \\ \hline
    $B_3(G)$ & 8 & $18 \leq B_3(G) \leq 21$ & $13 \leq B_3(G) \leq 18$  \\ \hline
    $T_3(G)$ & $12 = \#V(G)$ & $24=\#V(G)$ & 18 \\ \hline
    \end{tabular}
    \label{table1}
    \caption{Results for select Archimedean graphs}
    \end{table}

\subsection{The Flexible Tile Model} 
We work within the flexible tile model as described by Ellis-Monaghan et al. \cite{mintiles} and use the graph theoretical formalism of \cite{ellis2019tile}. For the convenience of the reader, we reproduce many of the essential definitions and results found in \cite{ellis2019tile} (see also \cite{jonoska2006spectrum}). At the basis of this model are $k$-armed branched junction molecules. A \emph{k-armed branched junction molecule} is a star-shaped molecule whose arms are formed from DNA strands.  At the end of each of these arms is what biologists often call a \emph{sticky end}, which is a set of unpaired nucleotide bases that can then bond to complementary ends on another molecule. For our purposes, we think of branched junction molecules as the simplest case, in which the arms are double stranded DNA with one strand extending beyond the other, as shown in Figure \ref{introfig1} (see \cite{SK94}).  This results in a relatively flexible molecule, and thus we assume no geometric restrictions with regard to angles or arm lengths. For a model that considers such restraints, see \cite{ferrari2018}. 

    \begin{figure}[hbt!]
        \centering
        \includegraphics[width=0.2\linewidth]{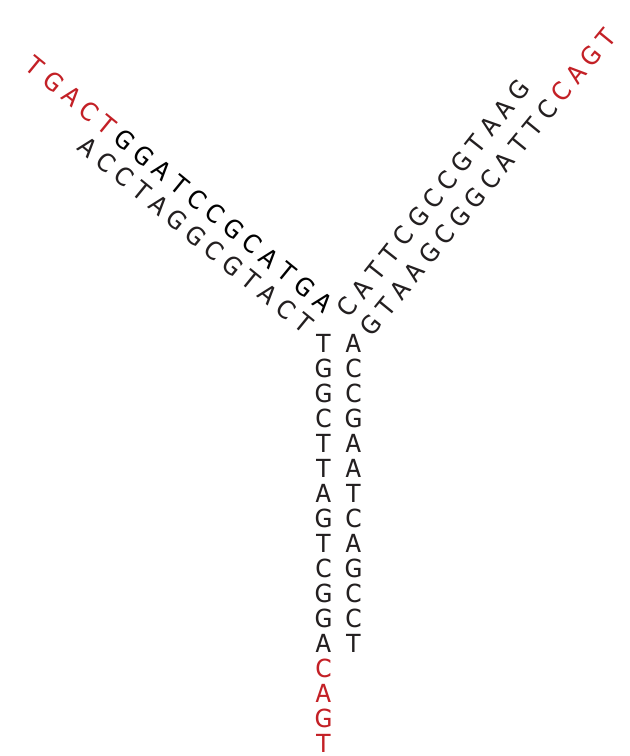}
        \caption{$k$-armed branched junction molecule}
        \label{introfig1}
    \end{figure}

The flexible tile model represents $k$-armed branched junction molecules as building blocks of discrete graphs. A discrete graph $G$ consists of a set $V=V(G)$ of vertices and a set $E=E(G)$ of edges together with a map $\mu:E \rightarrow V^{(2)}$ where $V^{(2)}$ is the set of unordered pairs of elements of $V$. In the context of this model, we allow for consideration of loop and multiple-edges. Throughout we use the notation $\#$ to denote cardinality, e.g. $\#V(G)$ will denote the order of $G$.

The region of unpaired bases at the end of an arm of a branched junction molecule is referred to as a \emph{cohesive-end}, and these various sequences of unpaired bases are represented with letter labels corresponding to half-edges incident with a vertex.

\begin{definition}
A \emph{cohesive-end type} is an element of a finite set $S=\Sigma \cup \hat{\Sigma} $, where $\Sigma$ is the set of hatted symbols and $\hat{\Sigma}$ is the set of un-hatted symbols.  Each cohesive-end type corresponds to a distinct arrangement of bases forming a cohesive-end on the end of a branched junction molecule arm, such that a hatted and an un-hatted symbol, say $a$ and $\hat{a}$, correspond to complementary cohesive-ends.   
\end{definition}

\begin{definition}
A cohesive-end type joined to its complement forms a \emph{bond-edge type}, which we identify by the un-hatted symbol, so for example, cohesive-end types $a$ and $\hat{a}$ will join to form a bond-edge of type $a$.
\end{definition}

The centers of $k$-armed branched junction molecules correspond to vertices, creating a representation of a molecule as a unit referred to as a \textit{tile}.

\begin{definition} A \emph{tile} is a graph-theoretical representation of a $k$-armed branched junction molecule as a vertex incident with $k$ half-edges (see Figure \ref{pot}). Tiles are denoted as multi-sets of cohesive-end types, with exponents indicating the number of that particular cohesive-end type.
\end{definition}
    
\begin{definition} A \emph{pot} $P$ is a set of tile types such that for each cohesive-end of type $h$ that appears in any tile $t_i \in P$, there exists a cohesive-end of type $\hat{h}$ (its complement) in some tile $t_j \in P$ (possibly $i = j$) and vice versa.
\end{definition}

    \begin{figure}[hbt!]
        \centering
        \includegraphics[width=0.7\linewidth]{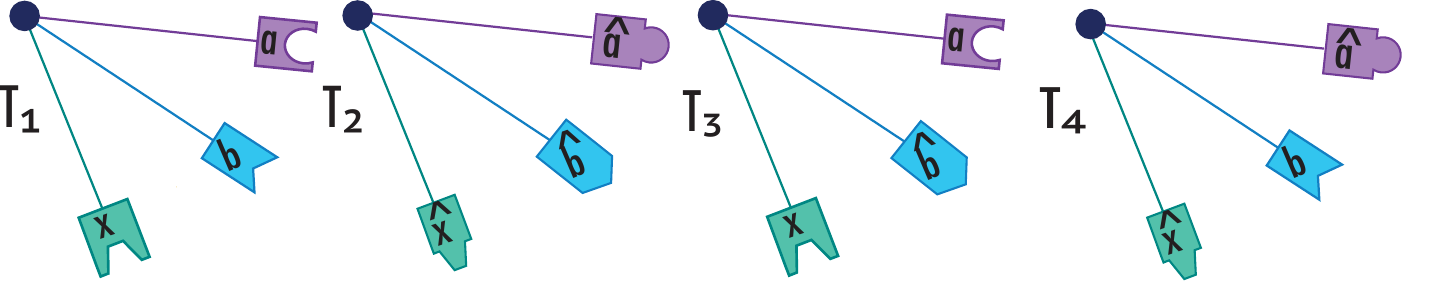}
        \caption{A pot of four tile types; $P = \{\{a,b,x\},\{\hat{a},\hat{b},\hat{x}\},\{a,\hat{b},x\},\{\hat{a},b,\hat{x}\}\}$.}
        \label{pot}
    \end{figure}

 We write $\#P$ to denote the number of distinct tile types in $P$. We use $\Sigma(P)$ to denote the set of bond-edge types that appear in tile types of $P$ and $\#\Sigma(P)$ to denote the number of such bond-edge types. We will often visually represent distinct bond-edge types with different colors; indeed, the theory of edge-coloring has been previously applied to this model (see \cite{BF2020}). If all vertices and edges of $G$ can be labeled using a subset of $P$, it is said that $P$ \emph{realizes} $G$. To configure $G$ from $P$, we utilize a method called \emph{assembly design}.

\begin{definition} An \emph{assembly design} of a graph $G$ is a labeling $\lambda$: $H \to \Sigma \cup \hat{\Sigma}$ of the half-edges of $G$ with the elements of $\Sigma$ and $\hat{\Sigma}$ such that if $e \in E(G)$ and $\mu(e) = \{u,v\}$, then $\widehat{\lambda(v,e)} = \lambda(u,e)$. This means that each edge of $G$ receives both a hatted and an un-hatted version of the same symbol, one on each of its half-edges. We use the convention that $\lambda$ provides each edge with an orientation that starts from the un-hatted half-edge to the hatted half-edge \cite{assemblydesign}.
\end{definition}

\begin{definition} The \emph{output of a pot $P$}, denoted $\mathcal{O}(P)$, is the set of all graphs realized by $P$. \end{definition}

While the graphs we study are undirected, an assembly design of a graph $G$ naturally creates an orientation for each edge of $G$. For brevity, we will write $\lambda(v,w)=a$ to denote that the edge $\{v,w\}$ is labeled using bond-edge type $a$, with cohesive-end type $a$ on the half-edge emanating from $v$ and $\hat{a}$ on the half-edge emanating from $w$.

For a given graph $G$, we consider three levels of restriction on a pot $P$ such that $G \in \mathcal{O}(P)$ \cite{mintiles}:
\begin{itemize}
    \item \emph{Scenario 1.} No restriction. Note that this allows the possibility that there exists $H \in \mathcal{O}(P)$ such that $\#V(H) < \#V(G)$. 
    \item \emph{Scenario 2.} The pot $P$ is such that for all $H \in \mathcal{O}(P)$, $\#V(H)\geq \#V(G)$. That is, no smaller order graphs may be realized by the pot, but graphs of equal order are allowable. 
    \item \emph{Scenario 3.} The pot $P$ is such that for all $H \in \mathcal{O}(P)$, $\#V(H)\geq\#V(G)$, and if $\#V(H)=\#V(G)$ then $H \cong G$. That is, no smaller-order or nonisomorphic graphs of equal order may be realized by the pot.
\end{itemize} 

We denote the \emph{minimum number of bond-edge types} and the \emph{minimum number of tile types} needed to realize $G$ in Scenario $i$ by $B_i(G)$ and $T_i(G)$, respectively.

\begin{example} The pot $P = \{\{a^3\},\{\hat{a}^3\},\{a^2,\hat{a}\},\{\hat{a}^2,a\}\}$ realizes $K_4$ as shown in Figure \ref{k4fig}. There exist graphs $H \in \mathcal{O}(P)$ such $\#V(H) < \#V(K_4)$ or such that $\#V(H)=\#V(K_4)$ but $H \not \cong K_4$. Examples are shown in Figure \ref{k4fig}. Thus, $P$ does not satisfy Scenario 2 (nor, necessarily, Scenario 3).

\begin{figure}[hbt!]
    \begin{minipage}{0.33\linewidth}
    \centering
    \begin{tikzpicture}
        \node[main node] (1) at (0,0) {$v_1$};
        \node[main node] (2) at (0,-1.7) {$v_2$};
        \node[main node] (3) at (-2,-3) {$v_3$};
        \node[main node] (4) at (2,-3) {$v_4$};
            \path[draw,thick]
            (1) edge node [near start,left]{$a$} (2)
            (2) edge node [near start,left]{$\hat{a}$} (1)
            (1) edge node [near start,left]{$a$} (3)
            (3) edge node [near start,left]{$\hat{a}$} (1)
            (1) edge node [near start,right]{$a$} (4)
            (4) edge node [near start,right]{$\hat{a}$} (1)
            (4) edge node [near start,below]{$\hat{a}$} (2)
            (2) edge node [near start,below]{$a$} (4)
            (4) edge node [near start,below]{$\hat{a}$} (3)
            (3) edge node [near start,below]{$a$} (4)
            (3) edge node [near start,below]{$a$} (2)
            (2) edge node [near start,below]{$\hat{a}$} (3);
    \end{tikzpicture}
\end{minipage}%
\begin{minipage}{0.33\linewidth}
    \centering
    \begin{tikzpicture}
        \node[main node] (a) at (-1.5,0) {$v_1$};
        \node[main node] (b) at (1.5,0) {$v_2$};
            \path[draw,thick]
            (a) edge node [near start,above]{$a$} (b)
            (b) edge node [near start,above]{$\hat{a}$} (a)
            [bend left=40] (a) edge node [near start,above]{$a$} (b)
            [bend right=40] (b) edge node [near start,above]{$\hat{a}$} (a)
            [bend right=40] (a) edge node [near start,above]{$a$} (b)
            [bend left=40] (b) edge node [near start,above]{$\hat{a}$} (a);
    \end{tikzpicture}
\end{minipage}%
\begin{minipage}{0.33\linewidth}
\centering
    \begin{tikzpicture}
        \node[main node] (1) at (-1,1) {$v_1$};
        \node[main node] (2) at (1,1) {$v_2$};
        \node[main node] (3) at (-1,-1) {$v_3$};
        \node[main node] (4) at (1,-1) {$v_4$};
            \path[draw,thick]
            (1) edge node [near start,below]{$a$} (2)
            (2) edge node [near start,below]{$\hat{a}$} (1)
            (1) edge node [near start,left]{$a$} (3)
            (3) edge node [near start,left]{$\hat{a}$} (1)
            (3) edge node [near start,above]{$a$} (4)
            (4) edge node [near start,above]{$\hat{a}$} (3)
            (4) edge node [near start,right]{$\hat{a}$} (2)
            (2) edge node [near start,right]{$a$} (4)
            [bend left=30] (1) edge node [near start,above]{$a$} (2)
            [bend right=30] (2) edge node [near start,above]{$\hat{a}$} (1)
            [bend right=30] (3) edge node [near start,below]{$a$} (4)
            [bend left=30] (4) edge node [near start,below]{$\hat{a}$} (3);
    \end{tikzpicture}
\end{minipage}
\caption{$K_4$ (left), smaller order graph (center), and nonisomorphic graph of equal order (right) as realized by $P = \{\{a^3\},\{\hat{a}^3\},\{a^2,\hat{a}\},\{\hat{a}^2,a\}\}$.}
\label{k4fig}
\end{figure}
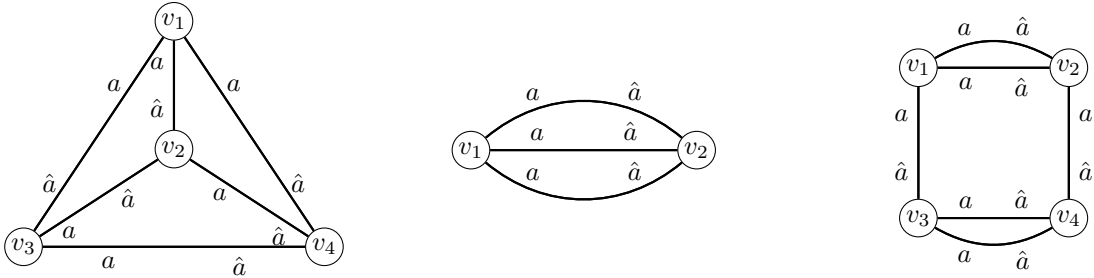
\end{example}

Bounds or exact values of $B_i(G)$ and $T_i(G)$ have already been determined for a variety of graph families, such as cycle graphs, complete graphs, trees, bipartite graphs, wheel graphs, gear graphs, and select sizes of lattice graphs \cite{almodovar2021complexity, mintiles, wheelgraphs, geargraphs}. Values for the skeletons of the Platonic solids, another group of $k$-regular polyhedral graphs, are computed in \cite{latticeandplatonic}. The skeletons of Arichmedean solids, referred to as the Archimedean graphs, are all 3, 4, or 5-regular. Hence, values for Scenario 1 are established in \cite{mintiles}. Specifically, $B_1(G)=1$ for all $k$-regular graphs, $T_1(G)=1$ for even values of $k$, and $T_1(G)=2$ for odd values of $k$. Little is known about the Archimedean graph family in Scenario 2. An upper bound for $B_2$ can be established for all Archimedean graphs by way of the result in \cite{BF2020} that $B_2(G) \leq \lceil n/2 \rceil + 1$ for all Hamiltonian graphs $G$ of order $n$. A pot besting this upper bound is given in \cite{ashworth} for the truncated tetrahedron, as it is shown that $B_2(G) \leq 3$ and $T_2(G)=4$ for this graph.

An example is given in \cite{almodovar2021complexity} to illustrate that there exist graphs $G$ for which the values of $B_3(G)$ and $T_3(G)$ cannot be achieved simultaneously by a pot, and in \cite{ferrari2022non} a pot that achieves both minimums is referred to as a \emph{biminimal} pot. In Sections \ref{unswap} and \ref{swap} we find bounds or exact values for $B_3$ and $T_3$ for the six Archimedean graphs of order 12 or 24. We provide exact values that correspond to biminimal pots for the cuboctahedron, snub cube, and truncated tetrahedron. For the small rhombicuboctahedron, truncated cube, and truncated octahedron, we only determine bounds for $B_3$. However, we note for the small rhombicuboctahedron and the truncated cube that any pot $P$ such that $\#\Sigma(P)=B_3$ will necessarily be a biminimal pot. 

\section{Methods}\label{sec:methods}
\subsection{Substructure Realization Problem}
Determining whether a pot that realizes a given graph $G$ realizes any graphs of order less than $\#V(G)$ has been referred to as the ``Substructure Realization Problem" (SRP) \cite{almodovar2021complexity}. To aid in addressing this problem, a system of equations and corresponding matrix, referred to as the \emph{construction matrix}, is defined in \cite{mintiles}; we reproduce those defintions here.
\begin{definition} Let $P$ be a pot with $p$ tiles labeled $t_1,$...$t_p$, and let $z_{i,j}$ be the net number of cohesive-ends of type $a_i$ on the tile $t_j$, with unhatted letters counting positively and hatted letters counting negatively. Define $r_i$ to be the proportion of tile type $t_i$ used in the assembly process. The following augmented matrix, referred to as the \emph{construction matrix} of $P$, encodes a system of equations that captures the necessary conditions for realization of a graph from $P$.

$$M(P) = \begin{bmatrix}
    z_{1,1} & z_{1,2} & \dots & z_{1,p} & 0\\
    \vdots & \vdots &   & \vdots\\
    z_{m,1} & z_{m,2} & \dots & z_{m,p} & 0\\
    1 & 1 & \dots & 1 & 1
    \end{bmatrix}$$
\end{definition}

\begin{definition} The \emph{spectrum of a pot $P$} is denoted $\mathcal{S}(P)$ and represents the intersection of the solution space of $M(P)$ with $\mathbb{Q}_{\geq 0}^p$, that is, the set of all $p-$tuples representing valid ratios of the $p$ tile types. \end{definition}

When $\mathcal{S}(P)$ consists of a unique solution $\langle r_1,...,r_p \rangle$, it is easy to determine the smallest order of a graph in $\mathcal{O}(P)$, given by the least common denominator of the $r_i$'s.  When there are dimensions of freedom in $\mathcal{S}(P)$, this problem can become more challenging.

\begin{example} The construction matrix and spectrum of the pot $P = \{\{a,b\},\{a,\hat{b}\},\{\hat{a}^2,b\},\{\hat{a}^2,\hat{b}\}\}$ are given below.
$$M(P) = \begin{bmatrix}
    1 & 1 & -2 & -2 & 0\\
    1 & -1 & 1 & -1 & 0\\
    1 & 1 & 1 & 1 & 1
    \end{bmatrix}$$
$$\mathcal{S}(P) = \left\{ \left\langle 1/6 + r_4, 1/2 -r_4, 1/3 - r_4, r_4 \right\rangle \mid r_4 \in \mathbb{Q}_{\geq 0} \right\}$$
For example, suppose $r_4 = 1/6$. Then the resulting tile proportions, $\langle 1/3, 1/3, 1/6, 1/6 \rangle$ demonstrate that there exists $H \in \mathcal{O}(P)$ with $\#V(H)=6$. Specifically, these tile proportions can be used to construct a $2 \times 3$ lattice graph.
\end{example}

 It is shown in \cite{almodovar2021complexity} that the SRP is, in general, NP-hard. In \cite{ashworth}, the authors reframe the SRP as an integer program, giving an efficient algorithm implemented in Python called the ``Substructure Realization Problem Solver" (SRPS), for determining the smallest order graph realized by a given pot regardless of the number of dimensions of freedom in $\mathcal{S}(P).$ The script in \cite{ashworth} is integral to our work with Archimedean graphs, as the optimal pots we find have as many as seven degrees of freedom in their spectrums. Rather than check the requirements of Scenario 2 by hand, we utilize the SRPS script to compute the minimum order of graphs realized by each proposed pot.

\subsection{Bounds for Unswappable and $k$-Regular Graphs}
Some general theory for $k$-regular graphs in Scenario 3 is provided in \cite{baek}; two of the primary results hinge on a property, coined therein, of some DNA self-assembly graphs called ``unswappable." In any scenario of the flexible tile model, incident edges $\{v,w\}, \{v,u\}$ may always be labeled with the same bond-edge type $b$ as long as the orientations of those edges are the same with regard to $v$ (that is, $\lambda(v,w)=\lambda(v,u)=b$ or $\lambda(v,w)=\lambda(v,u)=\hat{b}$). In this case, if the two edges break and re-bond in the alternate fashion, the graph is obviously unchanged. If $\lambda(v,w)=\lambda(u,v)$, then a loop edge may be realized, constituting an immediate violation of Scenario 3. Essentially, a graph $G$ is unswappable in the flexible tile model if any exchange of two \textit{non}-incident edges of $G$ (a breaking and re-joining of two bond-edges with the same bond-edge type such that the four half-edges bond in the alternate way) results in a graph nonisomorphic to $G$. The following formal definition and results are reproduced from \cite{baek}.

\begin{definition} We say that a graph $G$ is \emph{unswappable} if, for all pairs of disjoint edges $\{u, v\}$ and $\{s, t\}$, $G$ is not isomorphic to $G' = G - \{\{u, v\},\{s, t\}\} + \{\{u, t\}, \{s, v\}\}$. Otherwise, the graph is called \emph{swappable}.
\end{definition}

\begin{theorem}\label{thm:B3lower} Let $G$ be unswappable, and let $K$ be a minimum vertex cover of $G$. Then, $B_3(G) \geq |K|$. \end{theorem}

\begin{proposition}\label{thm:T3lower} Let $G$ be an unswappable graph and let $X$ be the number of distinct neighbor sets among all vertices (i.e. $X = |\{N(v): v\in V(G)\}|$). Then, $T_3(G) \geq X$. \end{proposition}

As noted in \cite{baek}, when labeling an unswappable graph in Scenario 3, each new bond-edge type introduced must stem from the same vertex, which we will refer to as a \emph{source vertex}. The following from \cite{baek} provide an upper bound for $B_3(G)$. Note that the following do not rely on a graph having the ``unswappable" property, but they do require the graph to be $k$-regular.

\begin{definition}\label{def:neighbind} Let $G$ be a graph and $M \subseteq V(G)$. We consider the vector space $\mathbb{R}^K$ represented as maps $K \to \mathbb{R}$. For each vertex $v \in V(G) - K$, we construct the vector $\rho_v$ by taking, for every $w \in K$:
$$ \rho_v(w)=\left\{
\begin{array}{cc}
        1 & \text{if } \{v,w\} \in E(G) \\
        0 & \text{else }
\end{array}
\right.$$
We say that $M$ is \emph{neighborhood independent} if the set of vectors $\{\rho_v\}_{v\in V(G)-K}$ is linearly independent.
\end{definition}

\begin{theorem}\label{thm:B3upper} Suppose that $G$ is a $k$-regular graph and $K \subseteq V(G)$ such that: $K$ is a vertex cover for $G$, $K$ is neighborhood independent, and the induced subgraph on $K$ is 2-edge-connected. Then $B_3(G) \le |K|$.   
\end{theorem}

To connect these bounds with a strategy for labeling such a $k$-regular graph, the following remark is also given in \cite{baek}.

\begin{remark}\label{labelremark} Suppose that $G$ is a $k$-regular graph. If we are given a neighborhood independent vertex cover $K \subseteq V(G)$ such that the induced subgraph on $K$ is 2-edge-connected, then we may construct a valid pot for $G$ using $|K|$ bond-edge types. In particular, this pot is the one derived from $(K, O)$, where $O$ is chosen to be any orientation on the induced subgraph of $K$.
\end{remark}

In Section \ref{unswap} we find that three of the six Archimedean graphs we study are unswappable and therefore apply Theorem \ref{thm:B3lower} and Proposition \ref{thm:T3lower}. In Section \ref{swap} we apply Theorem \ref{thm:B3upper} to the truncated cube, a graph we determine to be swappable.

\subsection{Computing Numbers of $n$-Cycles}
The structure of Archimedean graphs are such that each graph contains many $n$-cycles for various values of $n$. Thus, it is often the case that a labeling of these graphs results in a pot that realizes a graph with a different number of $n$-cycles for some small $n$.  There are well-known formulae for these values, given below, based on the adjacency matrix of a graph \cite{hararycycles}. The formula for the number of 5-cycles is simplified when the graph is $k$-regular, as are all Archimedean graphs. The formula for the number of 6-cycles in a graph is significantly more cumbersome to compute (see \cite{changcyles}), and so we did not incorporate this value in our methods.

\begin{theorem}\label{thm:3cycles} If $G$ is a simple graph with adjacency matrix $A$, then the number of 3-cycles in $G$ is $\frac{1}{6}tr(A^3).$ \end{theorem}

\begin{theorem}\label{thm:4cycles} If $G$ is a simple graph with adjacency matrix $A$, then the number of 4-cycles in $G$ is $\frac{1}{8}\left[ tr(A^4)-tr(A^2)-2\Sigma_{j \neq i} a_{ij}^{(2)}\right]$. \end{theorem}

\begin{theorem}\label{thm:5cycles} If $G$ is a simple graph with adjacency matrix $A$, then the number of 5-cycles in $G$ is $\frac{1}{10}\left[ tr(A^5)+5tr(A^3)-5\Sigma_{i=1}^n d_i a_{ij}^{(3)}\right]$, where $d_i$ is the degree of vertex $v_i$. \end{theorem}

In order to greatly expedite our work in Scenario 3, we composed a short script in Python \cite{cyclecode} to compute the number of 3-, 4-, and 5-cycles in a $k$-regular graph (substituting $d_i=k$ for all $i$ in Theorem \ref{thm:5cycles}) using the above formulae. Given an initial oriented edge ``labeled" with a bond-edge type, this script computes, for every other oriented edge in the graph, the adjacency matrix of the graph formed by breaking and re-bonding the initial edge with a second edge. The number of $n$-cycles for $n=3,4,5$ is computed for both the original graph and the graph formed by the change in edge connections. The script outputs a list of oriented edges for which the number of $n$-cycles differ. This process is outlined in Algorithm \ref{alg:cycles}. 

\begin{algorithm} 
\caption{Number of cycles algorithm}
\label{alg:cycles}
    \textbf{Input:} Adjacency matrix $A$ and $k$ for $k$-regular graph $G$ and initial edge $(v_i,v_j)$ such that $\lambda(v_i,v_j)=b$ \\
    \textbf{Output:} For $n=3,4,5$, a list and count of edges $(v_l,v_m)$ such that if $(v_i,v_j)$ and $(v_l,v_m)$ are swapped, then the number of $n$-cycles in $G$ has changed. \\
    \textbf{Algorithm:}
    \begin{algorithmic}[1]
        \State Using $A$, compute the number of $n$-cycles for $n=3,4,5$ for $G$ with the formulae from Theorems \ref{thm:3cycles}, \ref{thm:4cycles}, \ref{thm:5cycles}.
        \State For each $(l,m)$ such that $l, m \neq i, j$:
            \begin{enumerate} 
                \item Make a copy $A'$ of $A$ and alter $A'$ such that $A'$ is the adjacency matrix of the graph $G'$ formed after a swap of edges $(v_i,v_j)$ and $(v_l,v_m)$.
                \item Using $A'$, compute the number of $n$-cycles for $n=3,4,5$ for $G'$. If the number of $n$-cycles of $G'$ differs from that of $G$, append $(v_l,v_m)$ to the corresponding list.
            \end{enumerate}
        \State Return the lists, and their lengths, of edges altering the numbers of $n$-cycles.
    \end{algorithmic}
\end{algorithm}

In many cases, one can quickly verify by hand that an $n$-cycle is created where one did not previously exist. In other cases, it is quite difficult to see that the number of $n$-cycles for a given $n$ has changed, in which case the Python script is quite helpful. The number of edges in the Archimedean graphs can also make manual checks unpalatable; for example, the snub cube has 60 edges, with each having two potential labeling orientations. It should be noted that there are many other readily available scripts, often using a depth-first search method, that compute the number of cycles in a graph. Our script is specifically tailored to $k$-regular graphs in the flexible tile model. Throughout Sections \ref{unswap} and \ref{swap}, this script is used to determine which non-incident edges in a graph may not be labeled with the same bond-edge type.

\section{Unswappable Archimedean graphs}\label{unswap}
In this section, we determine that the cuboctahedron, the small rhombicubocathedron, and the snub cube are unswappable. We apply Theorem \ref{thm:B3lower} and Proposition \ref{thm:T3lower} to these graphs, obtaining the minimum value of $B_3$ for the cuboctahedron and snub cube, the minimum value of $T_3$ for all three graphs, and a close range of values of $B_3$ for the small rhombicuboctahedron. The graphs we show to be unswappable are 4- or 5-regular, and we hypothesize that the resulting increased connectivity is the driving factor behind the unswappability of the graphs.

Throughout, we refer to the vertex labeling of the cuboctahedron as shown in Figure \ref{fig:cuboct}.

\begin{lemma}\label{cuboctunswap} The cuboctahedron is unswappable.
\end{lemma}

\begin{proof} Note that the cuboctahedron is edge-transitive. Without loss of generality, suppose $\lambda(v_1,v_2)=b$. If any edge $(v_6,w)$, $(v_{10},w)$, $(v_3,w)$, $(w,v_3)$, $(w,v_4)$, or $(w,v_5)$, where $w$ is any other vertex, is labeled using bond-edge type $b$, then the resulting pot $P$ is such that there exists $H \in \mathcal{O}(P)$ with a multiple-edge (for example, see Figure \ref{fig:cuboct}). Together with the edges incident with $\{v_1,v_2\}$, this accounts for 26 of the 46 unlabeled oriented edges in $G$. For the remaining unlabeled oriented edges, it can be computationally verified that an edge ``swap" with $(v_1,v_2)$ results in a graph with a different number of $3$-cycles or a different number of $4$-cycles. 

Since the choice of oriented edge $(v_1,v_2)$ was arbitrary, we have shown that if $\lambda(v,w)=\lambda(v',w')=b$ with $v \neq v'$ and $w \neq w'$, then the resulting pot is such that a nonisomorphic graph can be realized.
\end{proof}

\begin{proposition}\label{cuboctB3} Let $G$ be the cuboctahedron. Then $B_3(G) = 8$. \end{proposition}

\begin{proof}
Note that the 4-cycles $\{v_2, v_3, v_{10}, v_{11}\}$ and $\{v_4, v_5, v_8, v_9\}$ are disjoint (share no common vertices or edges), so at least two non-adjacent vertices from each must be included in any minimum vertex cover of $G$. Without loss of generality, select $v_2$, $v_4$, $v_9$, and $v_{11}$. These four vertices cover 16 of the 24 edges of the graph. All remaining vertices can cover at most two remaining edges, so the order of a minimum vertex cover of $G$ is eight. By Theorem \ref{thm:B3lower}, $B_3(G) \ge 8$. Consider the minimum vertex cover $S = \{v_2, v_3, v_5, v_4, v_8, v_9, v_8, v_{10}\}$. Using $S$ as the set of source vertices for a labeling of $G$, we obtain the following pot $P$ such that $\#\Sigma(P)=8$ and $P$ realizes $G$ (see Figure \ref{fig:cuboct}).
\begin{align*} 
P = \{\{a^4\},\{b^4\},\{c^4\},\{d^4\},\{\hat{a},\hat{d},e^2\},\{\hat{a},\hat{b},\hat{e},\hat{f}\},\{\hat{c},\hat{d},\hat{e},\hat{f}\},\{\hat{b},\hat{c},f^2\}, \{\hat{a},\hat{d},g^2\}, \{\hat{a},\hat{b},\hat{g},\hat{h}\},\\\{\hat{c},\hat{d},\hat{g},\hat{h}\}, \{\hat{b},\hat{c},h^2\}\}
\end{align*}

It can be verified with the SRPS script from \cite{ashworth} that for all $H \in \mathcal{O}(P)$, $\#V(H) \geq 24 = \#V(G)$. Thus, $B_3(G)=8.$
\end{proof}

        \begin{figure}[hbt!]
        \centering
        \begin{tikzpicture}[transform shape, scale = 0.9]
            \node[main node] (4) at (0,0) {$v_4$};
            \node[main node] (5) at (2,0) {$v_5$};
            \node[main node] (8) at (0,-2) {$v_8$};
            \node[main node] (9) at (2,-2) {$v_9$};
            \node[main node] (1) at (1,2) {$v_1$};
            \node[main node] (6) at (-2,-1) {$v_6$};
            \node[main node] (7) at (4,-1) {$v_7$};
            \node[main node] (11) at (1,-4) {$v_{11}$};
            \node[main node] (2) at (-1,1) {$v_2$};
            \node[main node] (3) at (3,1) {$v_3$};
            \node[main node] (10) at (-1,-3) {$v_{10}$};
            \node[main node] (12) at (3,-3) {$v_{12}$};

            \path[draw,thick,color=red,->]
                (6) edge [bend left] node [near start, left]{$b$} (2)
                (1) edge [bend right] node [very near start, above]{$b$} (10);
             \path[draw,thick,color=black]
                (1) edge node []{} (3)
                (1) edge node []{} (4)
                (1) edge node []{} (5)
                (2) edge node []{} (6)
                (7) edge node []{} (3)
                (7) edge node []{} (5)
                (7) edge node []{} (9)
                (7) edge node []{} (12)
                (11) edge node []{} (10)
                (11) edge node []{} (8)
                (11) edge node []{} (9)
                (11) edge node []{} (12)
                (6) edge node []{} (4)
                (6) edge node []{} (8)
                (2) edge node []{} (3)
                (2) edge node []{} (10)
                (12) edge node []{} (3)
                (12) edge node []{} (10)
                (4) edge node []{} (8)
                (4) edge node []{} (5)    
                (9) edge node []{} (5)
                (9) edge node []{} (8);  
        \end{tikzpicture} \hspace{1cm}
        \begin{tikzpicture}[transform shape, scale = 0.9]
            \node[main node] (4) at (0,0) {$v_4$};
            \node[main node] (5) at (2,0) {$v_5$};
            \node[main node] (8) at (0,-2) {$v_8$};
            \node[main node] (9) at (2,-2) {$v_9$};
            \node[main node] (1) at (1,2) {$v_1$};
            \node[main node] (6) at (-2,-1) {$v_6$};
            \node[main node] (7) at (4,-1) {$v_7$};
            \node[main node] (11) at (1,-4) {$v_{11}$};
            \node[main node] (2) at (-1,1) {$v_2$};
            \node[main node] (3) at (3,1) {$v_3$};
            \node[main node] (10) at (-1,-3) {$v_{10}$};
            \node[main node] (12) at (3,-3) {$v_{12}$};

            \path[draw,thick,color=red,->]
                (1) edge node [near start, above]{$a$} (2)
                (1) edge node []{} (3)
                (1) edge node []{} (4)
                (1) edge node []{} (5);
             \path[draw,thick,color=blue,->]
                (7) edge node [near start, right]{$b$} (3)
                (7) edge node []{} (5)
                (7) edge node []{} (9)
                (7) edge node []{} (12);
             \path[draw,thick,color=green,->]
                (11) edge node [near start, below]{$c$} (10)
                (11) edge node []{} (8)
                (11) edge node []{} (9)
                (11) edge node []{} (12);
             \path[draw,thick,color=orange,->]
                (6) edge node [near start, left]{$d$} (2)
                (6) edge node []{} (4)
                (6) edge node []{} (8)
                (6) edge node []{} (10);
            \path[draw,thick,color=purple,->]
                (2) edge node [very near start, below]{$e$} (3)
                (2) edge node []{} (10);
            \path[draw,thick,color=teal,->]
                (12) edge node [near start, right]{$f$} (3)
                (12) edge node []{} (10);
            \path[draw,thick,color=pink,->]
                (4) edge node [near start, left]{$g$} (8)
                (4) edge node []{} (5);      
            \path[draw,thick,color=brown,->]
                (9) edge node [near start, left]{$h$} (5)
                (9) edge node []{} (8);  
        \end{tikzpicture}
            \caption{Nonisomorphic graph realized by pot when $\lambda(v_1,v_2)=\lambda(v_6,v_{10})$ (left); Scenario 3 labeling of cuboctahedron (right).}
            \label{fig:cuboct}
    \end{figure}
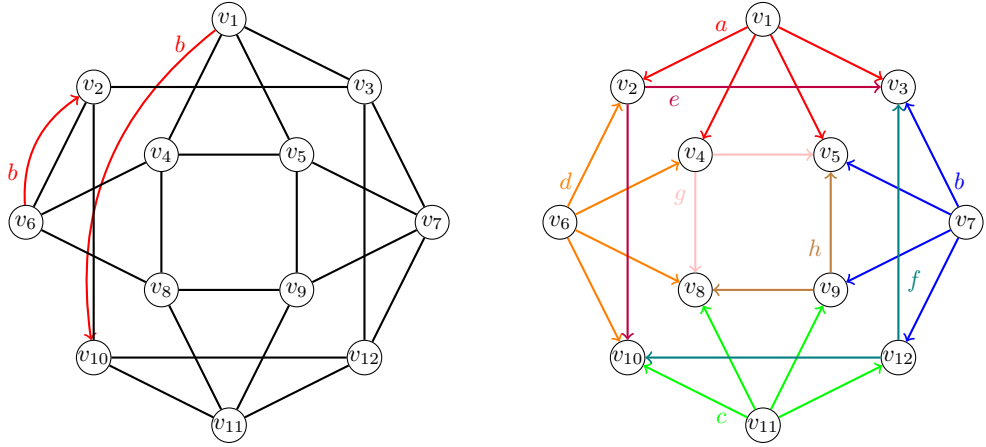

Throughout, we will consider the vertex numbering of the the small rhombicuboctahedron as shown in Figure \ref{fig:rhombi}.

\begin{lemma}\label{rhombiunswap} The small rhombicuboctahedron is unswappable. \end{lemma}

\begin{proof}
Note that, up to symmetry, there are two types of edges in $G$ -- those contained in only 4-cycles and those contained in both 3-cycles and 4-cycles. 

First, suppose $\lambda(v_1,v_2)=b$, noting that $\{v_1,v_2\}$ is not part of any 3-cycle. If any edge $(v_4,w)$, $(v_{8},w)$, $(v_{24},w)$, $(w,v_3)$, $(w,v_5)$, or $(w,v_{23})$ is labeled using bond-edge type $b$, then the resulting pot $P$ is such that there exists $H \in \mathcal{O}(P)$ with a multiple-edge. Together with the edges incident with $\{v_1,v_2\}$, this accounts for 20 of the 94 unlabeled oriented edges in $G$. For the remaining unlabeled oriented edges, it may be computationally verified that an edge ``swap" with $(v_1,v_2)$ results in a graph with a different number of $3$-cycles, $4$-cycles, or $5$-cycles. For nearly all such edges, if the edge is swapped with $(v_1,v_2)$, then the number of 3-cycles or 4-cycles changes. The three exceptions are $(v_6,v_9), (v_{11},v_{10})$, and $(v_{12},v_7)$, with which a swap changes the number of 5-cycles. Thus, when $\lambda(v_1,v_2)=b$, the only other edges that can be labeled with bond-edge type $b$ are those with initial vertex $v_1$ or terminal vertex $v_2$. 

Now suppose $\lambda(v_1,v_3)=b$, noting that $\{v_1,v_3\}$ is contained in a 3-cycle. In an identical fashion to the analysis of edge $(v_1,v_2)$, 20 of the 94 unlabeled oriented edges are either incident with $\{v_1,v_3\}$ or will produce a graph with a multiple-edge if swapped with $(v_1,v_3)$. For nearly all of the remaining unlabeled oriented edges, a swap results in a graph with a different number of $3$-cycles or a different number of $4$-cycles. The two exceptions are $(v_8, v_{18})$ and $(v_{12},v_8)$, with which a swap changes the number of 5-cycles. Thus, when $\lambda(v_1,v_3)=b$, the only other edges that can be labeled with bond-edge type $b$ are those with initial vertex $v_1$ or terminal vertex $v_3$. 

In both cases we have shown that if $\lambda(v,w)=\lambda(v',w')=b$ with $v \neq v'$ and $w \neq w'$, then the resulting pot is such that a nonisomorphic graph can be realized.
\end{proof}

In the proof of the following proposition, we apply the logic of Theorem \ref{thm:B3lower} to the small rhombicuboctahedron in order to create a Scenario 3 labeling of the graph using a minimum vertex cover as ``source" vertices. However, when using various minimum vertex covers of $G$, we find that smaller order graphs were realized by the resulting pots. Note that it is not feasible to check every possible minimum vertex cover for $G$, as for each 3-cycle there are three choices for a two-vertex covering set. We elect to slightly modify a pot derived from a minimum vertex cover to achieve a number of bond-edge types just one greater than the lower bound and leave it as an open question as to whether the minimum can be achieved.

\begin{proposition} \label{rhombiB3} Let $G$ be the small rhombicuboctahedron. Then, $16 \le B_3(G) \le 17$. \end{proposition}

\begin{proof} 
Since $G$ contains eight disjoint 3-cycles, and each 3-cycle has a minimum vertex cover of two, the cardinality of a minimum vertex cover of $G$ is 16. By Theorem \ref{thm:B3lower}, $B_3(G) \ge 16$. Consider the following minimum vertex cover of $G$.
$$S = \{v_2,v_3,v_4,v_5,v_6,v_9,v_{11},v_{12},v_{13},v_{15},v_{16},v_{17},v_{18},v_{19},v_{22},v_{23}\}$$ 
Using $S$ as the set of source vertices for a labeling of $G$ results in a pot $P$ such that $\#\Sigma(P)=16$ and $P$ realizes $G$. However, it can be verified with the SRPS script from \cite{ashworth} that there exists $H \in \mathcal{O}(P)$ with $\#V(H)=20 < \#V(G)$. Changing the labeling of a few edges to an additional bond-edge type gives the following pot of 17 tile types, which realizes $G$ as shown in Figure \ref{fig:rhombi}. 
\begin{align*} P' = \{\{a^4\},\{b^4\},\{c^4\},\{d^4\},\{e^4\},\{f^4\},\{g^4\},\{h^4\},\{\hat{a},\hat{b},i^2\},\{\hat{a},\hat{i},\hat{k},\hat{n}\},\{\hat{a},\hat{b},\hat{m},\hat{p}\},\{\hat{a},\hat{h},\hat{m},\hat{n}\}, \\ 
\{\hat{b},\hat{c},\hat{g},p\},\{\hat{b},\hat{c},j,q\},\{\hat{c},\hat{d},\hat{e},\hat{j}\},\{\hat{c},\hat{e},\hat{g},\hat{l}\},\{\hat{d},\hat{i},\hat{k},\hat{q}\},\{\hat{d},\hat{f},k^2\},\{\hat{d},\hat{e},\hat{f},o\},\{\hat{e},l^3\},\{\hat{f},\hat{h},\hat{l},\hat{o}\},\\ \{\hat{f},\hat{h},n^2\},\{\hat{g},m^3\},\{\hat{g},\hat{h},\hat{l},\hat{m}\}\} \end{align*}

Using the SRPS script from \cite{ashworth}, it can be verified that for all $H \in \mathcal{O}(P')$, $\#V(H) \geq 24 = \#V(G)$. Thus, $B_3(G) \le 17$. 
\end{proof}

\begin{figure}[hbt!]
    \centering
    \begin{tikzpicture}[transform shape,scale=0.8]
    \node[main node] (10) at (0,0) {$v_{10}$};
    \node[main node] (11) at (1,0) {$v_{11}$};
    \node[main node] (15) at (0,-1) {$v_{15}$};
    \node[main node] (16) at (1,-1) {$v_{16}$};
    \node[main node] (6) at (0,1) {$v_6$};
    \node[main node] (7) at (1,1) {$v_7$};
    \node[main node] (9) at (-1,0) {$v_9$};
    \node[main node] (12) at (2,0) {$v_{12}$};
    \node[main node] (14) at (-1,-1) {$v_{14}$};
    \node[main node] (17) at (2,-1) {$v_{17}$};
    \node[main node] (19) at (0,-2) {$v_{19}$};
    \node[main node] (20) at (1,-2) {$v_{20}$};
    \node[main node] (3) at (-0.5,2.5) {$v_3$};
    \node[main node] (4) at (1.5,2.5) {$v_4$};
    \node[main node] (5) at (-2.5,0.5) {$v_5$};
    \node[main node] (8) at (3.5,0.5) {$v_8$};
    \node[main node] (13) at (-2.5,-2.5) {$v_{13}$};
    \node[main node] (18) at (3.5,-1.5) {$v_{18}$};
    \node[main node] (21) at (-0.5,-3.5) {$v_{21}$};
    \node[main node] (22) at (1.5,-3.5) {$v_{22}$};
    \node[main node] (1) at (-3.5,3.5) {$v_1$};
    \node[main node] (2) at (4.5,3.5) {$v_2$};
    \node[main node] (23) at (-3.5,-4.5) {$v_{23}$};
    \node[main node] (24) at (4.5,-4.5) {$v_{24}$};

        \path[draw,thick,color=red,->]
        (2) edge node[]{} (24)
        (2) edge node[]{} (8)
        (2) edge node[very near start,above]{$a$} (1)
        (2) edge node[]{} (4);
        \path[draw,thick,color=blue,->]
        (3) edge node[very near start, above]{$b$} (1)
        (3) edge node[]{} (4)
        (3) edge node[]{} (5)
        (3) edge node[]{} (6);
        \path[draw,thick,color=green,->]
        (9) edge node[near start,above]{$c$} (5)
        (9) edge node[]{} (6)
        (9) edge node[]{} (14)
        (9) edge node[]{} (10);
        \path[draw,thick,color=orange,->]
        (11) edge node[near start, right]{$d$} (7)
        (11) edge node[]{} (12)
        (11) edge node[]{} (16)
        (11) edge node[]{} (10);
        \path[draw,thick,color=purple,->]
        (15) edge node[near start, right]{$e$} (10)
        (15) edge node[]{} (14)
        (15) edge node[]{} (19)
        (15) edge node[]{} (16);
        \path[draw,thick,color=teal,->]
        (17) edge node[near start,right]{$f$} (12)
        (17) edge node[]{} (16)
        (17) edge node[]{} (18)
        (17) edge node[]{} (20);
        \path[draw,thick,color=pink,->]
        (13) edge node[near start,left]{$g$} (5)
        (13) edge node[]{} (14)
        (13) edge node[]{} (21)
        (13) edge node[]{} (23);
        \path[draw,thick,color=brown,->]
        (22) edge node[near start, below]{$h$} (18)
        (22) edge node[]{} (20)
        (22) edge node[]{} (21)
        (22) edge node[]{} (24);
        \path[draw,thick,color=violet,->]
        (4) edge node[near start, below]{$i$} (7)
        (4) edge node[]{} (8);
        \path[draw,thick,color=magenta,->]
        (6) edge node[near start, right]{$j$} (10);
        \path[draw,thick,color=lime,->]
        (12) edge node[near start, above]{$k$} (7)
        (12) edge node[]{} (8);
        \path[draw,thick,color=cyan,->]
        (19) edge node[near start, below]{$l$} (14)
        (19) edge node[]{} (20)
        (19) edge node[]{} (21);
        \path[draw,thick,color=olive,->]
        (23) edge node[very near start, left]{$m$} (1)
        (23) edge node[]{} (24)
        (23) edge node[]{} (21);
        \path[draw,thick,color=Mahogany,->]
        (18) edge node[very near start, right]{$n$} (8)
        (18) edge node[]{} (24);
        \path[draw,thick,color=MidnightBlue,->]
        (16) edge node[near start, right]{$o$} (20);
        \path[draw,thick,color=yellow,->]
        (5) edge node[near start, right]{$p$} (1);
        \path[draw,thick,color=Tan,->]
        (6) edge node[very near start, above]{$q$} (7);
    \end{tikzpicture}
    \caption{Scenario 3 labeling of small rhombicuboctahedron}
    \label{fig:rhombi}
\end{figure}
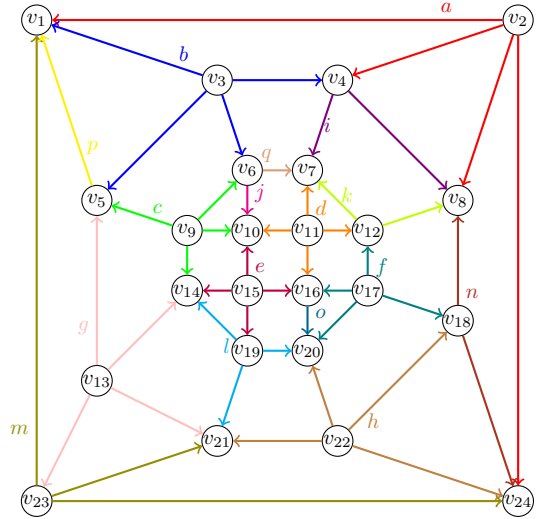

\begin{lemma}\label{snubcubeunswap} The snub cube is unswappable. \end{lemma}
\begin{proof}
Note that, up to symmetry, there are three types of edges in $G$: (1) those contained in 4-cycles, (2) those contained in two 3-cycles, each of which share an edge with a 4-cycle, (3) those contained in two 3-cycles, of which one shares an edge with a 4-cycle and one shares no edge with a 4-cycle. 

First, suppose $\lambda(v_1,v_2)=b$, noting that $\{v_1,v_2\}$ is contained in a 4-cycle. If any edge $(v_3,w)$, $(v_6,w)$, $(v_{15},w)$,$(v_{24},w)$, $(w,v_3)$, $(w,v_4)$, $(w,v_8)$, or $(w,v_{23})$ is labeled using bond-edge type $b$, then the resulting pot $P$ is such that there exists $H \in \mathcal{O}(P)$ with a multiple-edge. Together with the edges incident with $\{v_1,v_2\}$, this accounts for 43 of the 118 unlabeled oriented edges in $G$. For all but four of the remaining unlabeled oriented edges, it may be computationally verified that an edge ``swap" with $(v_1,v_2)$ results in a graph with a different number of $3$-cycles. A swap with edges $(v_{16}, v_9)$, $(v_{16},v_{19})$ $(v_{21},v_{18})$, or $(v_{22},v_{19})$ results in a graph with a different number of $4$-cycles. Thus, when $\lambda(v_1,\{v_1,v_2\})=b$, the only other edges that can be labeled with bond-edge type $b$ are those with initial vertex $v_1$ or terminal vertex $v_2$. 

Now, suppose $\lambda(v_1,v_3)=b$, noting that $\{v_1,v_3\}$ is contained in two 3-cycles that both share an edge with a 4-cycle. In an identical fashion to the analysis of edge $(v_1,v_2)$, 43 of the 118 unlabeled oriented edges are either incident with $\{v_1,v_3\}$ or will produce a graph with a multiple-edge if swapped with $(v_1,v_3)$. For all but two of the remaining unlabeled oriented edges, it may be computationally verified that an edge ``swap" with $(v_1,v_3)$ results in a graph with a different number of $3$-cycles. A swap with edges $(v_{7}, v_9)$ or $(v_{15},v_{24})$ results in a graph with a different number of $5$-cycles. Thus, when $\lambda(v_1,\{v_1,v_3\})=b$, the only other edges that can be labeled with bond-edge type $b$ are those with initial vertex $v_1$ or terminal vertex $v_3$. 

Now, suppose $\lambda(v_1,v_4)=b$, noting that $\{v_1,v_4\}$ is contained in two 3-cycles, one of which shares an edge with a 4-cycle. In an identical fashion to the analysis of edge $(v_1,v_2)$, 43 of the 118 unlabeled oriented edges are either incident with $\{v_1,v_4\}$ or will produce a graph with a multiple-edge if swapped with $(v_1,v_4)$. For all but two of the remaining unlabeled oriented edges, it may be computationally verified that an edge ``swap" with $(v_1,v_4)$ results in a graph with a different number of $3$-cycles. A swap with edge $(v_{16}, v_{19})$ results in a graph with a different number of $4$-cycles, and a swap with edge $(v_5,v_6)$ results in a graph with a different number of $5$-cycles. Thus, when $\lambda(v_1,\{v_1,v_4\})=b$, the only other edges that can be labeled with bond-edge type $b$ are those with initial vertex $v_1$ or terminal vertex $v_4$.

In all cases we have shown that if $\lambda(v,w)=\lambda(v',w')=b$ with $v \neq v'$ and $w \neq w'$, then the resulting pot is such that a nonisomorphic graph can be realized.
\end{proof}

\begin{proposition}\label{snubcubeB3} Let $G$ be the snub cube. Then $B_3(G)=16$. \end{proposition}
\begin{proof}
Since $G$ contains eight disjoint 3-cycles and each 3-cycle has a minimum vertex cover of two, the cardinality of a minimum vertex cover of $G$ is 16. By Theorem \ref{thm:B3lower}, $B_3(G) \ge 16$. Consider the following minimum vertex cover of $G$.
$$S = \{v_1,v_3,v_6,v_7,v_8,v_9,v_{10},v_{11},v_{14},v_{15},v_{16},v_{17},v_{18},v_{19},v_{22},v_{24}\}$$ 
Using $S$ as the set of source vertices for a labeling of $G$ results in the following pot $P$ such that $\#\Sigma(P)=16$ and $P$ realizes $G$ (see Figure \ref{fig:snubcube}). 
\begin{align*} P = \{\{a^5\},\{b^5\},\{c^5\},\{d^5\},\{e^5\},\{f^5\},\{g^5\},\{\hat{a},\hat{b},\hat{g},\hat{h},\hat{m}\},\{\hat{a},\hat{b},h^3\},\{\hat{a},\hat{c},\hat{h},\hat{i},\hat{j}\},\{\hat{a},\hat{f},i^3\}, \\ 
\{\hat{a},\hat{f},\hat{g},\hat{i},\hat{p}\},\{\hat{b},\hat{c},\hat{h},\hat{k},\hat{l}\},\{\hat{b},\hat{d},\hat{k},l^2\},\{\hat{b},\hat{d},\hat{g},m^2\},\{\hat{c},\hat{i},j^3\},\{\hat{c},k^4\},\{\hat{c},\hat{e},\hat{j},\hat{k},\hat{o}\}, \\
\{\hat{d},\hat{e},\hat{k},\hat{l},\hat{n}\},\{\hat{d},\hat{e},n^3\},\{\hat{d},\hat{g},\hat{m},\hat{n},\hat{p}\},\{\hat{e},\hat{f},\hat{j},o^2\},\{\hat{e},\hat{f},\hat{n},\hat{o},\hat{p}\},\{\hat{f},\hat{g},p^3\}\} \end{align*}

It can be verified with the SRPS script from \cite{ashworth} that $\#V(H) \geq 24 =  \#V(G)$ for all $H \in \mathcal{O}(P)$. Thus, $B_3(G)=16.$
\end{proof}

\begin{figure}[hbt!]
    \centering
    \begin{tikzpicture}[transform shape,scale=0.8]
        \node[main node] (10) at (0,0) {$v_{10}$};
        \node[main node] (12) at (-1,-0.5) {$v_{12}$};
        \node[main node] (17) at (-0.5,-1.5) {$v_{17}$};
        \node[main node] (13) at (0.5,-1) {$v_{13}$};
        \node[main node] (7) at (-1.25,0.75) {$v_7$};
        \node[main node] (11) at (1.25,0.25) {$v_{11}$};
        \node[main node] (18) at (0.75,-2.25) {$v_{18}$};
        \node[main node] (16) at (-1.75,-1.75) {$v_{16}$};
        \node[main node] (5) at (0.25,1.25) {$v_5$};
        \node[main node] (14) at (1.75,-0.75) {$v_{14}$};
        \node[main node] (20) at (-0.75,-2.75) {$v_{20}$};
        \node[main node] (9) at (-2.25,-0.25) {$v_9$};
        \node[main node] (4) at (-2.25,1.5) {$v_4$};
        \node[main node] (6) at (2,1.75) {$v_6$};
        \node[main node] (21) at (2.25,-2.5) {$v_{21}$};
        \node[main node] (19) at (-2.5,-3.25) {$v_{19}$};
        \node[main node] (8) at (-4.25,-0.75) {$v_8$};
        \node[main node] (22) at (0.25,-4.75) {$v_{22}$};
        \node[main node] (15) at (3.75,-0.25) {$v_{15}$};
        \node[main node] (3) at (-0.25,3.25) {$v_3$};
        \node[main node] (1) at (-5,4.25) {$v_1$};
        \node[main node] (2) at (4.75,4.25) {$v_2$};
        \node[main node] (23) at (-5.25,-5.75) {$v_{23}$};
        \node[main node] (24) at (4.75,-5.75) {$v_{24}$};

            \path[draw,thick, color=red, ->]
            (1) edge node[very near start,above]{$a$} (2)
            (1) edge node[]{} (3)
            (1) edge node[]{} (4)
            (1) edge node[]{} (8)
            (1) edge node[]{} (23);

             \path[draw,thick, color=blue, ->]
            (6) edge node[]{} (11)
            (6) edge node[]{} (15)
            (6) edge node[]{} (3)
            (6) edge node[]{} (5)
            (6) edge node[very near start,right]{$b$} (2);

            \path[draw,thick, color=green, ->]
            (7) edge node[]{} (4)
            (7) edge node[near start, above]{$c$} (5)
            (7) edge node[]{} (9)
            (7) edge node[]{} (10)
            (7) edge node[]{} (12);

            \path[draw,thick, color=orange, ->]
            (14) edge node[]{} (11)
            (14) edge node[]{} (13)
            (14) edge node[near start, above]{$d$} (15)
            (14) edge node[]{} (18)
            (14) edge node[]{} (21);

            \path[draw,thick, color=purple, ->]
            (17) edge node[]{} (12)
            (17) edge node[]{} (13)
            (17) edge node[]{} (16)
            (17) edge node[near start,below]{$e$} (18)
            (17) edge node[]{} (20);

            \path[draw,thick, color=teal, ->]
            (19) edge node[]{} (8)
            (19) edge node[]{} (16)
            (19) edge node[]{} (20)
            (19) edge node[]{} (22)
            (19) edge node[very near start,below]{$f$} (23);

            \path[draw,thick, color=pink, ->]
            (24) edge node[very near start, right]{$g$} (2)
            (24) edge node[]{} (15)
            (24) edge node[]{} (21)
            (24) edge node[]{} (22)
            (24) edge node[]{} (23);
            
            \path[draw,thick, color=brown, ->]
            (3) edge node[very near start, above]{$h$} (2)
            (3) edge node[]{} (4)
            (3) edge node[]{} (5);

            \path[draw,thick, color=violet, ->]
            (8) edge node[very near start, above]{$i$} (4)
            (8) edge node[]{} (9)
            (8) edge node[]{} (23);

            \path[draw,thick, color=magenta, ->]
            (9) edge node[]{} (4)
            (9) edge node[]{} (12) 
            (9) edge node[near start,left]{$j$} (16);

            \path[draw,thick, color=lime, ->]
            (10) edge node[]{} (5)
            (10) edge node[near start,above]{$k$} (11)
            (10) edge node[]{} (12)
            (10) edge node[]{} (13);

            \path[draw,thick, color=cyan, ->]
            (11) edge node[]{} (13)
            (11) edge node[near start, above]{$l$} (5);

            \path[draw,thick, color=olive, ->]
            (15) edge node[]{} (21)
            (15) edge node[very near start, right]{$m$} (2);

            \path[draw,thick, color=Mahogany, ->]
            (16) edge node[]{} (12)
            (16) edge node[near start,left]{$o$} (20);

            \path[draw,thick, color=MidnightBlue, ->]
            (18) edge node[]{} (13)
            (18) edge node[near start,below]{$n$} (20)
            (18) edge node[]{} (21);

            \path[draw,thick, color=yellow, ->]
            (22) edge node[near start, right]{$p$} (20)
            (22) edge node[]{} (21)
            (22) edge node[]{} (23);
    \end{tikzpicture}
    \caption{Scenario 3 labeling of the snub cube}
    \label{fig:snubcube}
\end{figure}
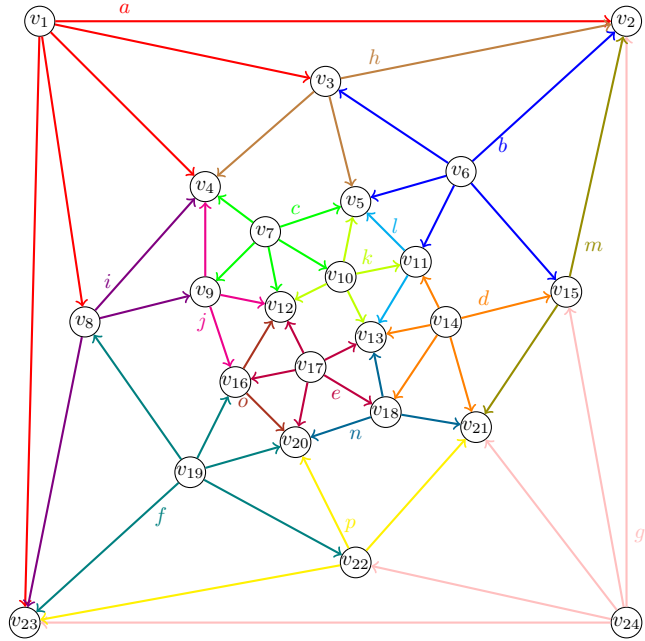

\begin{proposition} Let $G$ be the cuboctahedron, the small rhombicuboctahedron, or the snub cube. Then, $T_3(G) = \#V(G)$. \end{proposition}
\begin{proof}
We have shown in Lemmas \ref{cuboctunswap}, \ref{rhombiunswap}, and \ref{snubcubeunswap} that the cuboctahedron, small rhoombicubocathedron, and snub cube are unswappable. For each graph, it can be easily verified that each vertex has a distinct neighbor set. By Proposition \ref{thm:T3lower}, $T_3(G) = \#V(G)$. In each case, the example pots given in the proofs of \ref{cuboctB3}, \ref{rhombiB3}, and \ref{snubcubeB3} are such that $\#P = \#V(G)$.
\end{proof}

\section{Swappable Archimedean graphs}\label{swap}

In this section we find that three Archimedean graphs do not possess the ``unswappable" property, which renders the lower bounds from \cite{baek} inapplicable. The three graphs are all 3-regular, which may point to lesser connectivity and greater flexibility in labeling. Interestingly, we find that while the truncated tetrahedron and truncated cube are ``swappable," these graphs require no fewer bond-edge or tile types than if they were unswappable. That is, the swappability of the graph does not allow for greater repetition of bond-edge and tile types in the graph's labeling. On the other hand, we show that the truncated octahedron is a special case, such that fewer numbers of both tile and bond-edge types are sufficient. 

Throughout, we consider the vertex labeling of the truncated tetrahedron shown in Figure \ref{fig:trunctet}.

\begin{lemma}\label{trunctetrepeat} Let $G$ be the truncated tetrahedron. Then a single bond-edge type can be used to label at most three edges of $G$.
\end{lemma}

\begin{proof} 
 Note that, up to symmetry, there are two types of edges in $G$ -- those that are contained within a 3-cycle, and those that are contained only within 6-cycles. 
 
 First, suppose $\lambda(v_1,v_2)=b$, noting that $\{v_1,v_2\}$ is within a 3-cycle. If any edge $(v_3,w), (w,v_3)$, $(w,v_7)$, or $(v_11,w)$ is labeled using bond-edge type $b$, then the resulting pot $P$ is such that there exists a graph $H \in \mathcal{O}(P)$ with a multiple-edge. Together with the edges incident with $\{v_1,v_2\}$, this accounts for 14 of the 34 unlabeled oriented edges in $G$. If bond-edge type $b$ is used to label any of the edges $(v_4,v_5)$, $(v_5,v_6)$, $(v_6,v_4)$, $(v_6,v_5)$, $(v_5,v_8)$, $(v_8,v_5)$, $(v_7,v_8)$, $(v_7,v_9)$, $(v_9,v_8)$, $(v_6,v_{10})$, $(v_{10},v_6)$, $(v_{12},v_9)$, $(v_{10},v_{11})$, $(v_{10},v_{12})$, or $(v_{12},v_{11})$, then a graph with a 4-cycle can be realized by the resulting pot (for example, see Figure \ref{fig:trunctet}). If bond-edge type $b$ is used to label any of the edges $(v_4,v_6)$, $(v_5,v_4)$, $(v_8,v_9)$, $(v_9,v_{12})$, or $(v_{12},v_{10})$, then a graph with a $5$-cycle can be realized by the resulting pot (for example, see Figure \ref{fig:trunctet}). Thus, when $\lambda(v_1,v_2)=b$, the only other edges that can be labeled with bond-edge type $b$ are those with initial vertex $v_1$ or terminal vertex $v_2$.

Now, suppose $\lambda(v_1,v_7)=b$, noting that $\{v_1,v_7\}$ is not part of any 3-cycle. If any edge $(w,v_2), (w,v_3), (v_8,w)$, or $(v_9,w)$ is labeled using bond-edge type $b$, then there exists a graph $H \in \mathcal{O}(P)$ with a multiple-edge. If edges $\{v_4,v_6\},\{v_5,v_6\},\{v_6,v_{10}\}, \{v_{10},v_{11}\},$ $\{v_{10},v_{12}\}$ (in either orientation), $(v_5,v_4)$, or $(v_{12},v_{11})$ are labeled using bond-edge type $b$, then a graph with a 4-cycle can be realized by the resulting pot. If $(v_4,v_5)$ or $(v_{11},v_{12})$ are labeled using bond-edge type $b$, then a graph with a 5-cycle can be realized by the resulting pot. Only four oriented edges remain to be considered: $(v_3,v_4)$,$(v_2,v_{11})$,$(v_5,v_8)$, and $(v_{12},v_9)$. If either or both of the edges $(v_3,v_4)$ or $(v_2,v_{11})$ are labeled using bond-edge type $b$, then any resulting edge swaps result in a graph isomorphic to $G$. For example, consider the bijection $\phi(v_1)=v_3, \phi(v_3)=v_1, \phi(v_i)=v_i$ for $i \neq 4,6$ resulting from the ``swap" of edges $(v_1,v_7)$ and $(v_3,v_4)$. Labeling either or both of the edges $(v_5,v_8)$ or $(v_{12},v_9)$ with bond-edge type $b$ is an identical case. However, if $\lambda(v,w)=\lambda(v',w')=b$ where $v$ and $v'$ are not part of the same 3-cycle, then the resulting pot can realize a graph with a multiple-edge or a 4-cycle. For example, if $\lambda(v_3,v_4)=\lambda(v_5,v_8)$, a multiple-edge can be formed between $v_4$ and $v_5$. For an example of the other case, if $\lambda(v_3,v_4)=\lambda(v_{12},v_9)$, a 4-cycle can be formed by $\{v_4,v_6,v_{10},v_{12}\}$.

We have shown that a bond-edge type may be used to label at most three edges of $G$, in accordance with one of the following two conditions: (1) $\lambda(v,w) = \lambda(v,u) = \lambda(v,t)$, (2) $\lambda(v_1,w)=\lambda(v_2,u)=\lambda(v_3,t)$, where $\{v_1,v_2,v_3\}$ is a 3-cycle in $G$ and $w,u,t \not \in \{v_1,v_2,v_3\}$.
\end{proof}

\begin{proposition}\label{prop:trunctetb3} Let $G$ be the truncated tetrahedron. Then, $B_3(G) = 8$.
\end{proposition}

\begin{proof}
In accordance with the logic in the proof of Lemma \ref{trunctetrepeat}, a single bond-edge type can be used to label only one set of three non-incident edges. Without loss of generality, suppose $\{v_1,v_7\},\{v_3,v_4\},\{v_2,v_{11}\}$ are labeled with the same bond-edge type. For the edges that remain, a second bond-edge type can be used to label three edges only if all three are incident with the same vertex. Hence, we select an independent set of vertices to serve as source vertices for bond-edge types. A maximal set of independent vertices with no labeled incident edges is of cardinality three; for example, $\{v_5,v_9,v_{10}\}$. Thus, if a set of three non-incident edges are labeled using a single bond-edge type, then at most four bond-edge types can be used to label three edges each. 

An alternate labeling strategy is to use a single bond-edge type to label only edges incident with the same vertex. Note that a maximal independent set of vertices in $G$ is of cardinality four (for example, $\{v_1,v_5,v_9,v_{10}\}$). We may select four such vertices to be labeled with tile types consisting of a single cohesive-end type. After these edges are labeled, there remains no set of three edges that may be labeled using the same bond-edge type. 

Since no bond-edge type can be used to label more than three edges, if $G \in \mathcal{O}(P)$ and $\#\Sigma(P) \leq 7$, then at least five bond-edge types must be used to label three edges each. We have shown that at most four bond-edge types may be used to label three edges of $G$ each. Thus, $B_3(G) \geq 8$. The following pot $P$ is such that $\#\Sigma(P)=8$ and $P$ realizes $G$ (see Figure \ref{fig:trunctet}). 
$$P = \{\{a^3\},\{b^3\},\{c^3\},\{d^3\},\{e^2,\hat{a}\},\{f^2,\hat{c}\},\{g,\hat{b},\hat{d}\},\{h,\hat{d},\hat{e}\},\{\hat{a},\hat{b},\hat{e}\},\{\hat{b},\hat{f},\hat{g}\},\{\hat{a},\hat{c},\hat{f}\},\{\hat{c},\hat{d},\hat{h}\}\}$$

It can be computationally verified with the SRPS script from \cite{ashworth} that for all $H \in \mathcal{O}(P)$, $\#V(H) \geq 12 = \#V(G)$. Therefore, $B_3(G) = 8$.
\end{proof}

 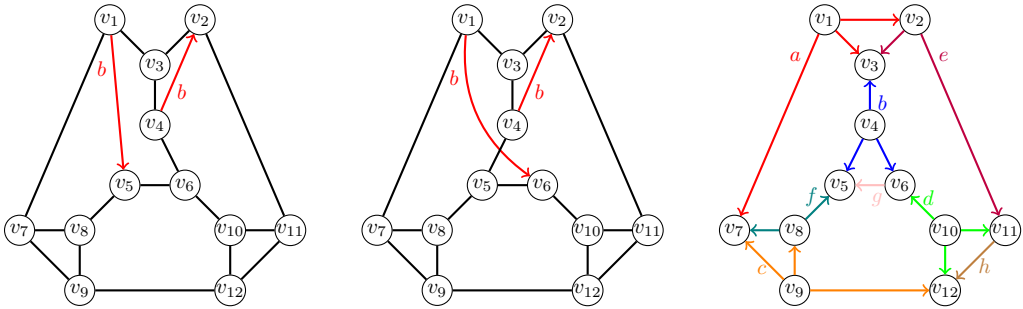
\begin{figure}[hbt!]
        \centering
        \begin{tikzpicture}[transform shape, scale = 0.8]
            \node[main node] (4) at (0,0.5) {$v_4$};
           \node[main node] (5) at (-0.5,-0.5){$v_5$};
          \node[main node] (6) at (0.5,-0.5){$v_6$};
         \node[main node] (3) at (0,1.5){$v_3$};
         \node[main node] (9) at (-1.25,-2.25){$v_9$};
         \node[main node] (12) at (1.25,-2.25){$v_{12}$};
         \node[main node] (1) at (-0.75,2.25){$v_1$};
         \node[main node] (2) at (0.75,2.25){$v_2$};
            \node[main node] (7) at (-2.25,-1.25){$v_7$};
            \node[main node] (11) at (2.25,-1.25){$v_{11}$};
          \node[main node] (8) at (-1.25,-1.25){$v_8$};
          \node[main node] (10) at (1.25,-1.25){$v_{10}$};

                \path[draw,thick,color=red,->]
                (1) edge node [near start, left]{$b$} (5)
                (4) edge node [near start, right]{$b$} (2);
                \path[draw,thick,color=black]
                (1) edge node []{} (3)
                (1) edge node []{} (7)
                (4) edge node []{} (3)
                (4) edge node []{} (6)
                (10) edge node []{} (6)
                (10) edge node []{} (11)
                (10) edge node []{} (12)
                (9) edge node []{} (7)
                (9) edge node []{} (8)
                (9) edge node []{} (12)
                (2) edge node []{} (3)
                (2) edge node []{} (11)
                (8) edge node []{} (7)
                (8) edge node []{} (5)
                (6) edge node []{} (5)
                (11) edge node []{} (12);
        \end{tikzpicture}
        \hspace{0.5cm}
        \begin{tikzpicture}[transform shape, scale = 0.8]
            \node[main node] (4) at (0,0.5) {$v_4$};
           \node[main node] (5) at (-0.5,-0.5){$v_5$};
          \node[main node] (6) at (0.5,-0.5){$v_6$};
         \node[main node] (3) at (0,1.5){$v_3$};
         \node[main node] (9) at (-1.25,-2.25){$v_9$};
         \node[main node] (12) at (1.25,-2.25){$v_{12}$};
         \node[main node] (1) at (-0.75,2.25){$v_1$};
         \node[main node] (2) at (0.75,2.25){$v_2$};
            \node[main node] (7) at (-2.25,-1.25){$v_7$};
            \node[main node] (11) at (2.25,-1.25){$v_{11}$};
          \node[main node] (8) at (-1.25,-1.25){$v_8$};
          \node[main node] (10) at (1.25,-1.25){$v_{10}$};
                \path[draw,thick,color=red,->]
                (1) edge [bend right] node [near start, left]{$b$} (6)
                (4) edge node [near start, right]{$b$} (2);
                \path[draw,thick,color=black]
                (1) edge node []{} (3)
                (1) edge node []{} (7)
                (4) edge node []{} (3)
                (4) edge node []{} (5)
                (10) edge node []{} (6)
                (10) edge node []{} (11)
                (10) edge node []{} (12)
                (9) edge node []{} (7)
                (9) edge node []{} (8)
                (9) edge node []{} (12)
                (2) edge node []{} (3)
                (2) edge node []{} (11)
                (8) edge node []{} (7)
                (8) edge node []{} (5)
                (6) edge node []{} (5)
                (11) edge node []{} (12);
        \end{tikzpicture}
        \hspace{0.5cm}
        \begin{tikzpicture}[transform shape, scale = 0.8]
            \node[main node] (4) at (0,0.5) {$v_4$};
           \node[main node] (5) at (-0.5,-0.5){$v_5$};
          \node[main node] (6) at (0.5,-0.5){$v_6$};
         \node[main node] (3) at (0,1.5){$v_3$};
         \node[main node] (9) at (-1.25,-2.25){$v_9$};
         \node[main node] (12) at (1.25,-2.25){$v_{12}$};
         \node[main node] (1) at (-0.75,2.25){$v_1$};
         \node[main node] (2) at (0.75,2.25){$v_2$};
            \node[main node] (7) at (-2.25,-1.25){$v_7$};
            \node[main node] (11) at (2.25,-1.25){$v_{11}$};
          \node[main node] (8) at (-1.25,-1.25){$v_8$};
          \node[main node] (10) at (1.25,-1.25){$v_{10}$};

                \path[draw,thick,color=red,->]
                (1) edge node []{} (2)
                (1) edge node []{} (3)
                (1) edge node [very near start, left]{$a$} (7);
                \path[draw,thick,color=blue,->]
                (4) edge node [near start, right]{$b$} (3)
                (4) edge node []{} (5)
                (4) edge node []{} (6);
                \path[draw,thick,color=green,->]
                (10) edge node [near start, above]{$d$} (6)
                (10) edge node []{} (11)
                (10) edge node []{} (12);
                \path[draw,thick,color=orange,->]
                (9) edge node [near start, left]{$c$} (7)
                (9) edge node []{} (8)
                (9) edge node []{} (12);
                 \path[draw,thick,color=purple,->]
                (2) edge node []{} (3)
                (2) edge node [very near start, right]{$e$} (11);
                \path[draw,thick,color=teal,->]
                (8) edge node []{} (7)
                (8) edge node [near start, above]{$f$} (5);
                \path[draw,thick,color=pink,->]
                (6) edge node [near start, below]{$g$} (5);
                \path[draw,thick,color=brown,->]
                (11) edge node [near start, below]{$h$} (12);
        \end{tikzpicture}
            \caption{nonisomorphic graphs realized by pot when $\lambda(v_1,v_2)=\lambda(v_4,v_5)$ (left) or $\lambda(v_1,v_2)=\lambda(v_4,v_6)$ (center); Scenario 3 labeling of truncated tetrahedron (right).}
            \label{fig:trunctet}
    \end{figure}

Throughout, we consider the vertex labeling of the truncated cube shown in Figure \ref{fig:trunccube}.

\begin{lemma}\label{trunccuberepeat} Let $G$ be the truncated cube. Then a single bond-edge type may be used to label at most three edges of $G$.
\end{lemma}

\begin{proof} 
Note that, up to symmetry, there are two types of edges $G$ -- those that are contained within a 3-cycle, and those that are contained only within 8-cycles. 

Suppose $\lambda(v_1,\{v_1,v_2\})=b$, noting that $\{v_1,v_2\}$ is not part of any 3-cycle. If any edge $(w,v_3),(w,v_4),(v_5,w)$, or $(v_6,w)$ is labeled using bond-edge type $b$, then the resulting pot $P$ is such that there exists $H \in \mathcal{O}(P)$ with a multiple-edge. Together with the edges incident with $\{v_1,v_2\}$, this accounts for 16 of the 70 unlabeled oriented edges in $G$. For 22 of the remaining unlabeled oriented edges, an edge swap with $(v_1,v_2)$ results in a graph with a 4-cycle (see Figure \ref{fig:trunccube}). If any edge $(v_{11},v_{13})$, $(v_{14},v_{12})$, $(v_{16},v_{15})$, $(v_{20},v_{17})$, or $(v_{18},v_{21})$ is swapped with $(v_1,v_2)$, the resulting graph contains a 6-cycle. Two remaining oriented edges, when swapped, will give rise to a nonplanar graph. If edge $(v_9,v_{10})$ is swapped with $(v_1,v_2)$, then a $K_{3,3}$ minor can be found in the vertex partitions $\{1,2,13\}, \{9,10,23\}$. If edge $(v_{23},v_{24})$ is swapped, then a $K_{3,3}$ minor can be found in the vertex partitions $\{1,2,16\}, \{20,23,24\}$. For 22 of the remaining oriented edges, it may be computationally verified that a swap with $(v_1,v_2)$ results in a graph with a different number of $3$-cycles. 

Five oriented edges remain to be considered: $(v_3,v_{19})$, $(v_4,v_7)$, $(v_8,v_5)$, $(v_{22},v_6)$, and $(v_{15},v_{16})$. If either of the edges $(v_3,v_{19}), (v_4,v_7)$ are labeled using bond-edge type $b$, then any resulting edge swaps result in a graph isomorphic to $G$. For example, consider the bijection $\phi(v_1)=v_3, \phi(v_3)=v_1$, $\phi(v_i)=v_i$ for $i \neq 1,3$ resulting from the swap of edges $(v_1,v_2)$ and $(v_3,v_{19})$ (one may visualize this as a ``twist" of the 3-cycle $\{v_1,v_3,v_4\}$). A swap with either of the edges $(v_8,v_5)$, $(v_{22},v_6)$ is a symmetric case. Finally, a swap of edges $(v_1,v_2)$ and $(v_{15},v_{16})$ results in a graph isomorphic to $G$ via the bijection $\phi(v_1) = v_{15}, \phi(v_{15})=v_1$, $\phi(v_3) = v_{17}, \phi(v_{17})=v_3$, $\phi(v_4)=v_{13}, \phi(v_{13})=v_4$, $\phi(v_7) = v_{11}, \phi(v_{11})=v_7$, $\phi(v_i)=v_i$ for all $i \not \in \{1,3,4,7,11,13,15,17\}$ (one may visualize this as a ``trading places" of the 3-cycles $\{v_1,v_3,v_4\}$ and $\{v_{13},v_{15},v_{17}\}$). 

We claim that, of the six edges $\{(v_1,v_2)$,$(v_3,v_{19})$, $(v_4,v_7)$, $(v_8,v_5)$, $(v_{22},v_6)$, and $(v_{15},v_{16})$, at most three can be labeled with the same bond-edge type. Swaps between any edge in $\{(v_3,v_{19}), (v_4,v_7)\}$ and any edge in $\{(v_8,v_5), (v_{22},v_6)\}$ result in a graph with a 4-cycle. A swap between edge $(v_{15},v_{16})$ and any edge in $\{(v_3,v_{19}), (v_4,v_7), (v_8,v_5), (v_{22},v_6)\}$ also results in a graph with a 4-cycle. Thus, the following are the ``maximal" sets including $(v_1,v_2)$ such that all edges in the set may be labeled with the same bond-edge type: $\{(v_1,v_2),(v_3,v_{19}), (v_4,v_7)\}$, $\{(v_1,v_2),(v_8,v_5),(v_{22},v_6)\}$, $\{(v_1,v_2),(v_{15},v_{16})\}$.

Now, suppose $\lambda(v_1,\{v_1,v_3\})=b$, noting that $\{v_1,v_3\}$ is within a 3-cycle. If any edge $(w,v_2)$, $(v_4,w)$, $(w,v_4)$, or $(v_{19},w)$ is labeled using bond-edge type $b$, then resulting the pot $P$ is such that there exists $H \in \mathcal{O}(P)$ with a multiple-edge. Together with the edges incident with $\{v_1,v_3\}$, this accounts for 14 of the 70 unlabeled oriented edges of $G$. There are 45 edges for which, if labeled with bond-edge type $b$, a swap with edge $\{v_1,v_3\}$ results in a graph with a 4-cycle or 5-cycle. For the remaining 11 unlabeled oriented edges, it may be computationally verified that a swap with edge $(v_1,v_2)$ results in a graph with a different number of $3$-cycles. Thus, when $\lambda(v_1,v_3)=b$, the only other edges that can be labeled with bond-edge type $b$ are those with initial vertex $v_1$ or terminal vertex $v_3$.

We have shown that a bond-edge type may be used to label at most three edges of $G$, in accordance with one of the following two conditions: (1) $\lambda(v,w) = \lambda(v,u) = \lambda(v,t)$, (2) $\lambda(v_1,w)=\lambda(v_2,u)=\lambda(v_3,t)$, where $\{v_1,v_2,v_3\}$ is a 3-cycle in $G$ and $w,u,t \not \in \{v_1,v_2,v_3\}$.
\end{proof}

\begin{proposition}\label{prop:trunccubeB3} Let $G$ be the truncated cube. Then, $16 \leq B_3(G) \leq 21$. 
\end{proposition}

\begin{proof}
It is shown in the proof of Lemma \ref{trunccuberepeat} that at least two bond-edge types are needed to label each 3-cycle in $G$, and that those bond-edge types must be distinct for each 3-cycle. Therefore, at least 16 bond-edge types are needed to label the 24 edges of the 3-cycles in $G$. By carefully selecting labeling orientations, it is possible to label all remaining 12 edges with bond-edge types used to label the edges of the 3-cycles while maintaining a pot that does not realize any nonisomorphic graphs. However, upon testing such pots in the SRPS script from \cite{ashworth}, graphs of smaller order were consistently realized. 

Consider the following vertex cover of $G$.
$$K = \{v_1, v_2, v_3, v_4, v_5, v_6, v_7, v_8, v_{10}, v_{11}, v_{12}, v_{13}, v_{14}, v_{16}, v_{17}, v_{18}, v_{20}, v_{21}, v_{22}, v_{23},  v_{24}\}$$
The 21 vertices of $K$ form a 2-edge-connected subgraph, as shown by the bolded edges of $G$ in Figure \ref{fig:trunccube}. Then, following from Definition \ref{def:neighbind}, the corresponding vectors for the three vertices in $G \setminus S$ are given below.
\begin{align*} \rho_{v_9} = \langle 0,0,0,0,0,0,1,0,1,1,0,0,0,0,0,0,0,0,0,0,0 \rangle \\
\rho_{v_{15}} = \langle 0,0,0,0,0,0,0,0,0,0,0,1,0,1,1,0,0,0,0,0,0 \rangle \\
\rho_{v_{19}} = \langle 0,0,1,0,0,0,0,0,0,0,0,0,0,0,0,0,1,0,0,1,0 \rangle \end{align*}

 As it can be easily verified that $\{\rho_{v_9},\rho_{v_{15}},\rho_{v_{19}}\}$ is a linearly independent set, $K$ is neighborhood independent. By Theorem \ref{thm:B3upper}, $B_3(G) \le |K| = 21$. The following pot $P$, constructed according to Remark \ref{labelremark}, is such that $\#\Sigma(P)=21$ and $P$ realizes $G$ (see Figure \ref{fig:trunccube}).
\begin{align*} P = \{\{a,\hat{m},\hat{u}\},\{\hat{a},b^2\},\{\hat{b},c^2\},\{\hat{c},d^2\},\{\hat{d},e^2\},\{\hat{e},f^2\},\{\hat{f},g^2\},\{\hat{g},h^2\},
    \{\hat{h},i^2\},\{\hat{i},j^2\}, \{\hat{i},\hat{j},k\}, \\ \{\hat{k},l^2\},\{\hat{l},m^2\},\{\hat{l},\hat{m},n\},\{\hat{n},o^2\},\{\hat{o},p^2\},\{\hat{o},\hat{p},q\},\{\hat{q},r^2\},\{\hat{r},s^2\},\{\hat{r},\hat{s},t\},\{\hat{b},\hat{t},u\},\{\hat{c},\hat{d},\hat{j}\}, \\ \{\hat{e},\hat{f},\hat{s}\},\{\hat{g},\hat{h},\hat{p}\}\}\end{align*}

It can be computationally verified with the SRPS script from \cite{ashworth} that for all $H \in \mathcal{O}(P)$, $\#V(H) \geq 24 = \#V(G)$. Therefore, $B_3(G) \le 21$.
\end{proof}

    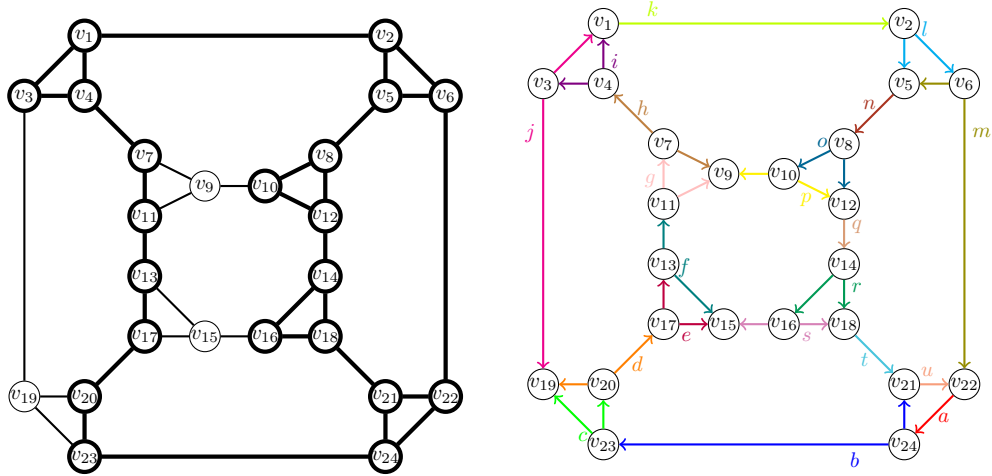
\begin{figure}[hbt!]
        \centering
        \centering
        \begin{tikzpicture}[transform shape,scale=0.8]
            \node[main node] (9) at (-0.5,1) {$v_9$};
            \node[main node, ultra thick] (10) at (0.5,1) {$v_{10}$};
            \node[main node, ultra thick] (11) at (-1.5,0.5) {$v_{11}$};
            \node[main node, ultra thick] (12) at (1.5,0.5) {$v_{12}$};
            \node[main node, ultra thick] (13) at (-1.5,-0.5) {$v_{13}$};
            \node[main node, ultra thick] (14) at (1.5,-0.5) {$v_{14}$};
            \node[main node] (15) at (-0.5,-1.5) {$v_{15}$};
            \node[main node, ultra thick] (16) at (0.5,-1.5) {$v_{16}$};
            \node[main node, ultra thick] (7) at (-1.5,1.5) {$v_7$};
            \node[main node, ultra thick] (8) at (1.5,1.5) {$v_8$};
            \node[main node, ultra thick] (17) at (-1.5,-1.5) {$v_{17}$};
            \node[main node, ultra thick] (18) at (1.5,-1.5) {$v_{18}$};
            \node[main node, ultra thick] (4) at (-2.5,2.5) {$v_4$};
            \node[main node, ultra thick] (5) at (2.5,2.5) {$v_5$};
            \node[main node, ultra thick] (20) at (-2.5,-2.5) {$v_{20}$};
            \node[main node, ultra thick] (21) at (2.5,-2.5) {$v_{21}$};
            \node[main node, ultra thick] (1) at (-2.5,3.5) {$v_1$};
            \node[main node, ultra thick] (2) at (2.5,3.5) {$v_2$};
            \node[main node, ultra thick] (23) at (-2.5,-3.5) {$v_{23}$};
            \node[main node, ultra thick] (24) at (2.5,-3.5) {$v_{24}$};
            \node[main node, ultra thick] (3) at (-3.5,2.5) {$v_3$};
            \node[main node, ultra thick] (6) at (3.5,2.5) {$v_6$};
            \node[main node] (19) at (-3.5,-2.5) {$v_{19}$};
            \node[main node, ultra thick] (22) at (3.5,-2.5) {$v_{22}$};

                \path[draw,ultra thick]
                (22) edge node []{} (24)
                (22) edge node []{} (6)
                (22) edge node []{} (21)
                (24) edge node []{} (21)
                (24) edge node []{} (23)
                (23) edge node []{} (20)
                (21) edge node []{} (18)
                (20) edge node []{} (17)
                (18) edge node []{} (14)
                (18) edge node []{} (16)
                (17) edge node []{} (13)
                (16) edge node []{} (14)
                (14) edge node []{} (12)
                (13) edge node []{} (11)
                (12) edge node []{} (8)
                (12) edge node []{} (10)
                (11) edge node []{} (7)
                (10) edge node []{} (8)
                (8) edge node []{} (5)
                (7) edge node []{} (4)
                (6) edge node []{} (2)
                (6) edge node []{} (5)
                (5) edge node []{} (2)
                (4) edge node []{} (1)
                (4) edge node []{} (3)
                (3) edge node []{} (1)
                (2) edge node []{} (1);
            \path[draw,thick]
                (3) edge node []{} (19)
                (19) edge node []{} (20)
                (19) edge node []{} (23)
                (7) edge node []{} (9)
                (9) edge node []{} (11)
                (13) edge node []{} (15)
                (15) edge node []{} (17)
                (15) edge node []{} (16)
                (9) edge node []{} (10);
            
    \end{tikzpicture}
    \hspace{0.5cm}
        \begin{tikzpicture}[transform shape,scale=0.8]
            \node[main node] (9) at (-0.5,1) {$v_9$};
            \node[main node] (10) at (0.5,1) {$v_{10}$};
            \node[main node] (11) at (-1.5,0.5) {$v_{11}$};
            \node[main node] (12) at (1.5,0.5) {$v_{12}$};
            \node[main node] (13) at (-1.5,-0.5) {$v_{13}$};
            \node[main node] (14) at (1.5,-0.5) {$v_{14}$};
            \node[main node] (15) at (-0.5,-1.5) {$v_{15}$};
            \node[main node] (16) at (0.5,-1.5) {$v_{16}$};
            \node[main node] (7) at (-1.5,1.5) {$v_7$};
            \node[main node] (8) at (1.5,1.5) {$v_8$};
            \node[main node] (17) at (-1.5,-1.5) {$v_{17}$};
            \node[main node] (18) at (1.5,-1.5) {$v_{18}$};
            \node[main node] (4) at (-2.5,2.5) {$v_4$};
            \node[main node] (5) at (2.5,2.5) {$v_5$};
            \node[main node] (20) at (-2.5,-2.5) {$v_{20}$};
            \node[main node] (21) at (2.5,-2.5) {$v_{21}$};
            \node[main node] (1) at (-2.5,3.5) {$v_1$};
            \node[main node] (2) at (2.5,3.5) {$v_2$};
            \node[main node] (23) at (-2.5,-3.5) {$v_{23}$};
            \node[main node] (24) at (2.5,-3.5) {$v_{24}$};
            \node[main node] (3) at (-3.5,2.5) {$v_3$};
            \node[main node] (6) at (3.5,2.5) {$v_6$};
            \node[main node] (19) at (-3.5,-2.5) {$v_{19}$};
            \node[main node] (22) at (3.5,-2.5) {$v_{22}$};

                \path[draw,thick, color=red, ->]
                (22) edge node [near start,below]{$a$} (24);

                \path[draw, thick, color=blue, ->]
                (24) edge node []{} (21)
                (24) edge node [very near start,below]{$b$} (23);

                \path[draw, thick, color=green, ->]
                (23) edge node [near start,below]{$c$} (19)
                (23) edge node []{} (20);

                \path[draw, thick, color=orange, ->]
                (20) edge node []{} (19)
                (20) edge node [near start,right]{$d$} (17);

                \path[draw, thick, color=purple, ->]
                (17) edge node [near start,below]{$e$} (15)
                (17) edge node []{} (13);

                \path[draw, thick, color=teal, ->]
                (13) edge node [near start,above]{$f$} (15)
                (13) edge node []{} (11);

                \path[draw, thick, color=pink, ->]
                (11) edge node [near start,left]{$g$} (7)
                (11) edge node []{} (9);

                \path[draw, thick, color=brown, ->]
                (7) edge node [near start,above]{$h$} (4)
                (7) edge node []{} (9);

                \path[draw, thick, color=violet, ->]
                (4) edge node [near start,right]{$i$} (1)
                (4) edge node []{} (3);

                \path[draw, thick, color=magenta, ->]
                (3) edge node []{} (1)
                (3) edge node [very near start,left]{$j$} (19);

                \path[draw, thick, color=lime, ->]
                (1) edge node [very near start,above]{$k$} (2);

                \path[draw, thick, color=cyan, ->]
                (2) edge node []{} (5)
                (2) edge node [near start,above]{$l$} (6);

                \path[draw, thick, color=olive, ->]
                (6) edge node []{} (5)
                (6) edge node [very near start,right]{$m$} (22);

                \path[draw, thick, color=Mahogany, ->]
                (5) edge node [near start,left]{$n$} (8);

                \path[draw, thick, color=MidnightBlue, ->]
                (8) edge node [near start,above]{$o$} (10)
                (8) edge node []{} (12);

                \path[draw, thick, color=yellow, ->]
                (10) edge node []{} (9)
                (10) edge node [near start,below]{$p$} (12);

                \path[draw, thick, color=Tan, ->]
                (12) edge node [near start,right]{$q$} (14);

                \path[draw, thick, color=ForestGreen, ->]
                (14) edge node [near start,right]{$r$} (18)
                (14) edge node []{} (16);

                \path[draw, thick, color=Thistle, ->]
                (16) edge node []{} (15)
                (16) edge node [near start,below]{$s$} (18);

                \path[draw, thick, color=SkyBlue, ->]
                (18) edge node [near start,below]{$t$} (21);

                \path[draw, thick, color=Melon, ->]
                (21) edge node [near start,above]{$u$} (22);
    \end{tikzpicture}
        \caption{2-edge-connected subgraph formed by $K$ (left); Scenario 3 labeling of truncated cube (right).}
        \label{fig:trunccube}
\end{figure}

\begin{proposition}\label{prop:tetcubeT3} Let $G$ be the truncated tetrahedron or the truncated cube. Then, $T_3(G) = \#V(G)$.   
\end{proposition}

\begin{proof} 
The proof of Theorem 3 in \cite{mintiles}, which states that the number of tile types is bounded below by the chromatic number of $G$, shows that tile types may not repeated on adjacent vertices in any graph. Suppose $\lambda(v)=t$, where $v$ is any vertex of $G$ and $t$ is an arbitrary tile type. If $\lambda(w)=t$ for any other vertex of $G$ that is non-adjacent to $v$, then at least one edge the 3-cycle containing $w$ will be labeled with a bond-edge type used to label an edge of the 3-cycle containing $t$. As shown in the proofs of Lemmas \ref{trunctetrepeat} and \ref{trunccuberepeat}, this is not permissible for a labeling of $G$ in Scenario 3. Since each vertex of $G$ requires a unique tile type, $T_3(G)=\#V(G)$. The example pots given in the proofs of Propositions \ref{prop:trunctetb3} and \ref{prop:trunccubeB3} are such that $\#P=\#V(G)$ and $\#P=24$.
\end{proof}

As mentioned in the introduction, the truncated octahedron has a special biological application of cell-targeting. Its structure closely resembles that of oxidized low-density lipoprotein (ox-LDL), allowing it to bind to the low-density lipoprotein receptor-1 (LOX-1) on cells \cite{applications}. This property could lend itself to targeting cancerous cells and delivering chemotherapy drugs encaged in a DNA truncated octahedron nanostructure. The truncated octahedron also turns out to be especially mathematically interesting as a DNA graph, as it is the only smaller order Archimedean graph for which the number of tile types required in Scenario 3 is less than the order of the graph. 

\begin{lemma}\label{truncoctbond} Let $G$ be the truncated octahedron. Then a single bond-edge type may be used to label at most three edges of $G$. \end{lemma}

\begin{proof}
Note that, up to symmetry, there are two types of edges $G$ -- those that are contained within a 4-cycle, and those that are contained only within 6-cycles. Suppose $\lambda(v_1,\{v_1,v_2\})=b$, noting that $\{v_1,v_2\}$ is contained within a 4-cycle. If any edge $(w,v_3),(w,v_{23}),(v_6,w)$, or $(v_{24},w)$ is labeled using bond-edge type $b$, then the resulting pot $P$ is such that there exists $H \in \mathcal{O}(P)$ with a multiple-edge. Together with the edges incident with $\{v_1,v_2\}$, this accounts for 16 of the 70 unlabeled oriented edges in $G$. For 13 of the remaining unlabeled oriented edges, a swap with $(v_1,v_2)$ results in a graph with a 3-cycle (see Figure \ref{fig:trunccube}). Consider the edges $(v_4,v_8)$, $(v_5,v_9)$, $(v_7,v_{15})$, $(v_{10},v_{18})$, $(v_{16},v_{15})$, $(v_{18},v_{17})$, $(v_{19},v_{15})$, and $(v_{18},v_{22})$. Due to the symmetry of $G$, it is sufficient to only consider the edges $(v_4,v_8),(v_7,v_{15}),(v_{16},v_{16}),$ and $(v_{19},v_{15})$. A swap of $(v_1,v_2)$ with $(v_4,v_8)$, $(v_{16},v_{15})$, or $(v_{19},v_{15})$ results in a nonplanar graph; the vertex partitions $\{\{v_1,v_2,v_7\},\{v_3,v_8,v_{23}\}\}$, $\{\{v_2,v_{22},v_{23}\},\{v_6,v_{16},v_{24}\}\}$, $\{\{v_2,v_{22},v_{23}\},\{v_5,v_{19},v_{24}\}\}$ form a subdivison of $K_{3,3}$ in the graph after a swap with $(v_4,v_8)$, $(v_{16},v_{15})$, and $(v_{19},v_{15})$, respectively (an example is illustrated in Figure \ref{fig:truncoct}). If edges $(v_1,v_2)$ and $(v_7,v_{15})$ are swapped, then the edge $(v_{15},v_{19})$ will be contained within two 4-cycles, and such an edge does not exist in $G$. For the 33 remaining unlabeled oriented edges it may be computationally verified that a swap with edge $(v_1,v_2)$ results in a graph with a different number of 4-cycles. Thus, when $\lambda(v_1,v_2)=b$, the only other edges that can be labeled with bond-edge type $b$ are those with initial vertex $v_1$ or terminal vertex $v_2$.

Now, suppose $\lambda(v_1,v_3)=v$, noting that $\{v_1,v_3\}$ is not contained within any 4-cycle. In an identical fashion to the analysis of edge $(v_1,v_2)$, 16 of the 70 unlabeld oriented edges are either incident with $\{v_1,v_3\}$ or will produce a graph with a multiple-edge if swapped with $(v_1,v_3)$. For 12 of the remaining unlabeled oriented edges, a swap with $(v_1,v_3)$ results in a graph with a 3-cycle. If any of the edges $\{v_{10}, v_{18}\}$, $\{v_{12}, v_9\}$, $\{v_{13}, v_{16}\}$, $\{v_{20}, v_{21}\}$ are labeled with bond-edge type $a$, then there exists a graph $H \in \mathcal{O}(P)$ containing a 5-cycle. Consider the edges $(v_2, v_6)$, $(v_5, v_4)$, $(v_6, v_{10})$, $(v_9, v_5)$, $(v_{14}, v_{17})$, $(v_{15}, v_7)$, $(v_{16}, v_{15})$, $(v_{17}, v_{14})$, $(v_{19}, v_{20})$, and $(v_{23}, v_{19})$. Due to the symmetry of $G$, it is sufficient to only consider the edges $(v_{14}, v_{17})$, $(v_{15}, v_7)$, $(v_{16}, v_{15})$, $(v_{17}, v_{14})$, $(v_{19}, v_{20})$, and $(v_{23}, v_{19})$. A swap of $(v_1,v_3)$ with $(v_{14},v_{17})$, $(v_{15},v_7)$, $(v_{16},v_{15})$, $(v_{17},v_{14})$, $(v_{19},v_{20})$, or $(v_{23}, v_{19})$ results in a nonplanar graph ($\{\{v_1,v_9,v_{24}\},\{v_2,v_{14},v_{17}\}\}$, $\{\{v_1,v_3,v_{11}\},$ $\{v_5,v_7,v_{15}\}\}$, $\{\{v_3,v_8,v_{15}\},$ $\{v_4,v_7,v_{16}\}\}$, $\{\{v_1,v_3,v_{10}\},$ $\{v_6,v_{14},v_{17}\}\}$, $\{\{v_1,v_3,v_{24}\},$ $\{v_2,v_{19},v_{20}\}\}$, and $\{\{v_1,v_3,v_{24}\},$ $\{v_2,v_{19},v_{23}\}\}$ are example vertex partitions forming a subdivison of $K_{3,3}$ in the graph resulting from a swap with $(v_{14},v_{17})$,$(v_{15},v_7)$, $(v_{16},v_{15})$, $(v_{17},v_{14})$, $(v_{19},v_{20})$, and $(v_{23}, v_{19})$, respectively). For 26 of the remaining unlabeled oriented edges it may be computationally verified that a swap with edge $(v_1,v_3)$ results in a graph with a different number of 4-cycles. If edges $(v_{11},v_8)$, or $(v_{24},v_{22})$ are swapped with $(v_1,v_3),$ the resulting graph is isomorphic to $G$ via the bijections induced by $\phi(v_3)=v_8$, and $\phi(v_3)=v_{22}$, respectively. (2) $\lambda(v_1,w)=\lambda(w,v_2)$, where $v_1,v_2$ are non-adjacent vertices contained within the same 4-cycle.

We have shown that a bond-edge type may be used to label at most three edges of $G$, in accordance with the condition $\lambda(v,w)=\lambda(v,u)=\lambda(v,t)$. 

Note that labeling either edge incident with $(v_1,v_3)$ and either edge $(v_{11},v_8)$, $(v_{24},v_{22})$ with the same bond-edge type will result in a labeling allowing for a nonisomorphic graph, as shown previously by the analysis beginning with $(v_1,v_2)$. In addition, only two of the $(v_1,v_3), (v_{11},v_8), (v_{24},v_{22})$ may be simultaneously labeled with the same bond-edge type; otherwise, a swap of edges $(v_{11},v_8)$ and $(v_{24},v_{22})$ will result in a nonplanar graph (in the same fashion described earlier in this proof). 

In all cases, we have shown that a maximum of three edges of $G$ may be labeled with the same bond-edge type.
\end{proof}

\begin{proposition}\label{prop:truncoctB3} Let $G$ be the truncated octahedron. Then, $13 \leq B_3(G) \le 18$.
\end{proposition}

\begin{proof}  
As shown in Lemma \ref{truncoctbond}, a single bond-edge type may be used to label a maximum of three edges of $G$. Since $G$ has 36 edges, $B_3(G) \geq 12$. Following the logic of the proof of Lemma \ref{truncoctbond}, a pot $P$ such that $\#\Sigma(P)=12$ and $P$ realizes $G$ may only be obtained by using the same bond-edge type to label all three edges incident with a given vertex. To achieve this labeling, select a maximal independent set of vertices of $G$ (which must consist of two vertices from each 4-cycle) to serve as source vertices. This results in a labeling of $G$ using twelve bond-edge types. However, such a labeling results in a pot $P$ such that there exists $H \in \mathcal{O}(P)$ with $\#V(H)=6$ (consider $H$ realized by three copies of the tile type corresponding to a non-source vertex together with one copy each of the tile types corresponding the three adjacent source vertices). Thus, twelve bond-edge types is insufficient for labeling $G$ in Scenario 3, and $B_3(G) \geq 13.$

The following pot $P$, derived from using each bond-edge type to label two edges of $G$ in accordance with the conditions outlined in the proof of Lemma \ref{truncoctbond}, is such that $\#\Sigma(P)=18$ and $P$ realizes $G$ (see Figure \ref{fig:truncoct}).

\begin{align*} P = \{\{a,b,c\},\{d,e,f\},\{g,h,i\},\{j,k,l\},\{m,n,o\},\{p,q,r\}, \{\hat{a}^2,i\},\{\hat{c},\hat{e}^2\},\{\hat{f},\hat{g}^2\},\{\hat{d}^2,\hat{l}\}, \{\hat{h}^2,\hat{r}\}, \\ \{\hat{i},\hat{j}^2\}, \{\hat{k}^2,\hat{o}\},\{\hat{f},\hat{m}^2\},\{\hat{l},\hat{p}^2\},\{\hat{n}^2,\hat{r}\},\{\hat{c},\hat{q}^2\},\{\hat{b}^2,\hat{o}\}\} \end{align*}

\end{proof}

\begin{proposition}\label{prop:truncoctT3}
Let $G$ be the truncated octahedron. Then, $T_3(G) = 18$. 
\end{proposition}

\begin{proof} 
If a tile type is used to label any vertex $v$, then by the logic of the proof of Lemma \ref{truncoctbond}, that tile type may only be used to label the vertex non-adjacent to $v$ within the same 4-cycle; otherwise, at least one bond-edge type would be used to label two non-incident edges in different 4-cycles. Note that within any 4-cycle it is not possible to use only two tile types, as once one pair of non-adjacent vertices of the cycle are labeled with the same tile type, the other pair of non-adjacent vertices cannot possibly be labeled to match. Thus each 4-cycle requires at least three tile types, and $T_3(G) \geq 18$. The pot given in the proof of Proposition \ref{cuboctB3} is such that $\#P=18$.
\end{proof}

        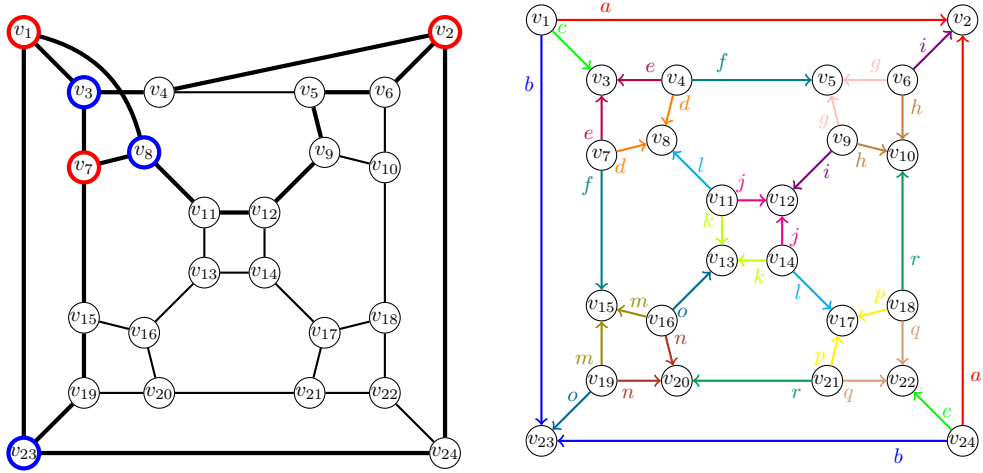
\begin{figure}[hbt!]
        \centering
        \begin{tikzpicture}[transform shape, scale = 0.8]
            \node[main node] (11) at (0,0){$v_{11}$};
            \node[main node] (12) at (1,0){$v_{12}$};
            \node[main node] (13) at (0,-1){$v_{13}$};
            \node[main node] (14) at (1,-1){$v_{14}$};
            \node[main node,draw=blue, ultra thick] (8) at (-1,1){$v_8$};
            \node[main node] (9) at (2,1){$v_9$};
            \node[main node] (16) at (-1,-2){$v_{16}$};
            \node[main node] (17) at (2,-2){$v_{17}$};
            \node[main node,draw=blue, ultra thick] (3) at (-2,2){$v_3$};
            \node[main node] (6) at (3,2){$v_6$};
            \node[main node] (19) at (-2,-3){$v_{19}$};
            \node[main node] (22) at (3,-3){$v_{22}$};
            \node[main node, draw=red, ultra thick] (1) at (-3,3){$v_1$};
            \node[main node, draw=red, ultra thick] (2) at (4,3){$v_2$};
            \node[main node,draw=blue, ultra thick] (23) at (-3,-4){$v_{23}$};
            \node[main node] (24) at (4,-4){$v_{24}$};
            \node[main node] (4) at (-0.75,2){$v_4$};
            \node[main node] (5) at (1.75,2){$v_5$};
            \node[main node] (20) at (-0.75,-3){$v_{20}$};
            \node[main node] (21) at (1.75,-3){$v_{21}$};
            \node[main node, draw=red, ultra thick] (7) at (-2,0.75){$v_7$};
            \node[main node] (15) at (-2,-1.75){$v_{15}$};
            \node[main node] (10) at (3,0.75){$v_{10}$};
            \node[main node] (18) at (3,-1.75){$v_{18}$};

            \path[draw,ultra thick]
            (1) edge node []{} (23)
            (1) edge node []{} (3)
            (4) edge node []{} (3)
            (7) edge node []{} (3)
            (11) edge node []{} (8)
            (11) edge node []{} (12)
            (9) edge node []{} (12)
            (9) edge node []{} (5)
            (6) edge node []{} (5)
            (6) edge node []{} (2)
            (24) edge node []{} (2)
            (24) edge node []{} (23)
            (19) edge node []{} (23)
            (19) edge node []{} (15)
            (7) edge node []{} (15)
            (7) edge node []{} (8)

            (1) edge [bend left] node []{} (8)
            (4) edge node []{} (2);
            
            \path[draw,thick]
            (24) edge node []{} (22)
            (4) edge node []{} (5)
            (6) edge node []{} (10)
            (9) edge node []{} (10)
            (14) edge node []{} (12)
            (14) edge node []{} (13)
            (11) edge node []{} (13)
            (14) edge node []{} (17)
            (16) edge node []{} (15)
            (19) edge node []{} (20)
            (16) edge node []{} (20)
            (16) edge node []{} (13)
            (18) edge node []{} (17)
            (21) edge node []{} (17)
            (18) edge node []{} (22)
            (21) edge node []{} (22)
            (18) edge node []{} (10)
            (21) edge node []{} (20);
        \end{tikzpicture}
        \hspace{0.5cm}
        \begin{tikzpicture}[transform shape, scale = 0.8]
            \node[main node] (11) at (0,0){$v_{11}$};
            \node[main node] (12) at (1,0){$v_{12}$};
            \node[main node] (13) at (0,-1){$v_{13}$};
            \node[main node] (14) at (1,-1){$v_{14}$};
            \node[main node] (8) at (-1,1){$v_8$};
            \node[main node] (9) at (2,1){$v_9$};
            \node[main node] (16) at (-1,-2){$v_{16}$};
            \node[main node] (17) at (2,-2){$v_{17}$};
            \node[main node] (3) at (-2,2){$v_3$};
            \node[main node] (6) at (3,2){$v_6$};
            \node[main node] (19) at (-2,-3){$v_{19}$};
            \node[main node] (22) at (3,-3){$v_{22}$};
            \node[main node] (1) at (-3,3){$v_1$};
            \node[main node] (2) at (4,3){$v_2$};
            \node[main node] (23) at (-3,-4){$v_{23}$};
            \node[main node] (24) at (4,-4){$v_{24}$};
            \node[main node] (4) at (-0.75,2){$v_4$};
            \node[main node] (5) at (1.75,2){$v_5$};
            \node[main node] (20) at (-0.75,-3){$v_{20}$};
            \node[main node] (21) at (1.75,-3){$v_{21}$};
            \node[main node] (7) at (-2,0.75){$v_7$};
            \node[main node] (15) at (-2,-1.75){$v_{15}$};
            \node[main node] (10) at (3,0.75){$v_{10}$};
            \node[main node] (18) at (3,-1.75){$v_{18}$};
            
            \path[draw,thick, color=red, ->]
            (1) edge node [very near start, above]{$a$} (2)
            (24) edge node [very near start, right]{$a$} (2);

            \path[draw,thick, color=blue, ->]
            (1) edge node [very near start, left]{$b$} (23)
            (24) edge node [very near start, below]{$b$} (23);

            \path[draw,thick, color=green, ->]
            (1) edge node [near start, above]{$c$} (3)
            (24) edge node [very near start, above]{$c$} (22);

            \path[draw,thick, color=orange, ->]
            (4) edge node [near start, right]{$d$} (8)
            (7) edge node [very near start, below]{$d$} (8);

            \path[draw,thick, color=purple, ->]
            (4) edge node [near start, above]{$e$} (3)
            (7) edge node [very near start, left]{$e$} (3);

            \path[draw,thick, color=teal, ->]
            (4) edge node [near start, above]{$f$} (5)
            (7) edge node [very near start, left]{$f$} (15);

            \path[draw,thick, color=pink, ->]
            (6) edge node [near start, above]{$g$} (5)
            (9) edge node [very near start, left]{$g$} (5);

            \path[draw,thick, color=brown, ->]
            (6) edge node [near start, right]{$h$} (10)
            (9) edge node [very near start, below]{$h$} (10);

            \path[draw,thick, color=violet, ->]
            (6) edge node [near start, above]{$i$} (2)
            (9) edge node [very near start, below]{$i$} (12);

            \path[draw,thick, color=magenta, ->]
            (14) edge node [near start, right]{$j$} (12)
            (11) edge node [very near start, above]{$j$} (12);

            \path[draw,thick, color=lime, ->]
            (14) edge node [near start, below]{$k$} (13)
            (11) edge node [very near start, left]{$k$} (13);

            \path[draw,thick, color=cyan, ->]
            (11) edge node [near start, above]{$l$} (8)
            (14) edge node [very near start, below]{$l$} (17);
            
            \path[draw,thick, color=olive, ->]
            (16) edge node [near start, above]{$m$} (15)
            (19) edge node [very near start, left]{$m$} (15);

            \path[draw,thick, color=Mahogany, ->]
            (19) edge node [near start, below]{$n$} (20)
            (16) edge node [very near start, right]{$n$} (20);

            \path[draw,thick, color=MidnightBlue, ->]
            (16) edge node [near start, below]{$o$} (13)
            (19) edge node [very near start, left]{$o$} (23);

            \path[draw,thick, color=yellow, ->]
            (18) edge node [near start, above]{$p$} (17)
            (21) edge node [near start, left]{$p$} (17);

            \path[draw,thick, color=Tan, ->]
            (18) edge node [near start, right]{$q$} (22)
            (21) edge node [very near start, below]{$q$} (22);

            \path[draw,thick, color=ForestGreen, ->]
            (18) edge node [near start, right]{$r$} (10)
            (21) edge node [very near start, below]{$r$} (20);
        \end{tikzpicture}
            \caption{Subdivision of $K_{3,3}$ minor formed when $\lambda(v_1,v_2)=\lambda(v_4,v_8)$ (left); Scenario 3 labeling of truncated octahedron (right).}
            \label{fig:truncoct}
    \end{figure}

As we have only determined a range of values for $B_3$ for the truncated octahedron, and the pot $P$ achieving $T_3$ is such that $\#\Sigma(P)$ significantly exceeds the lower bound for $B_3$, we conjecture that the truncated octahedron is another example of a graph for which a biminimal pot does not exist.

\section{Conclusion}
 While all graphs called for the maximum number of tile types in Scenario 3, the minimization of $B_3$, particularly for the order 12 graphs, means that the DNA can be cut with fewer different types of restriction enzymes, which can be notoriously expensive. Thus, even in the strictest conditions, the cost of producing Archimedean solid nanostructures could be lessened. Given the biological applications of polyhedral structures discussed in the introduction, if it is cheaper to assemble them, there is a possibility that related therapies and treatments could become more widely accessible. One direction for future work is finding pots to achieve lower bounds for bond-edge types or prove such pots do not exist for the small rhombicuboctahedron, the truncated cube, and the truncated octahedron. Furthermore, it would be interesting to explore whether a biminimal pot exists for the truncated octahedron. Finally, the Archimedean graphs of order greater than 24 have yet to be considered.

\bibliography{bibliography}

\begin{bibdiv}
\begin{biblist}

\bib{almodovar2021complexity}{article}{
      author={Almod{\'o}var, Leyda},
      author={Ellis-Monaghan, Jo},
      author={Harsy, Amanda},
      author={Johnson, Cory},
      author={Sorrells, Jessica},
       title={Computational complexity and pragmatic solutions for flexible tile based {{DNA}} self-assembly},
        date={2024},
     journal={Natural Computing},
       pages={1\ndash 22},
}

\bib{latticeandplatonic}{misc}{
      author={Almodóvar, Leyda},
      author={Ellis-Monaghan, Joanna},
      author={Harsy, Amanda},
      author={Johnson, Cory},
      author={Sorrells, Jessica},
       title={Optimal tile-based dna self-assembly designs for lattice graphs and platonic solids},
        date={2021},
}

\bib{ashworth}{misc}{
      author={Ashworth, Jacob},
      author={Grossmann, Luca},
      author={Navarro, Fausto},
      author={Almodóvar, Leyda},
      author={Harsy, Amanda},
      author={Johnson, Cory},
      author={Sorrells, Jessica},
       title={Algorithmic pot generation: Algorithms for the flexible-tile model of {DNA} self-assembly},
        date={2024},
        note={https://www.arxiv.org/abs/arXiv:2408.00192},
}

\bib{baek}{misc}{
      author={Baek, Lisa},
      author={Bove, Ethan},
      author={Cho, Michael},
      author={Zhang, Xingyi},
       title={Optimal constructions for {DNA} self-assembly of $k$-regular graphs},
        date={2025},
        note={https://www.arxiv.org/abs/2502.03716},
}

\bib{BF2020}{article}{
      author={Bonvicini, Simona},
      author={Ferrari, Margherita~Maria},
       title={On the minimum number of bond-edge types and tile types: an approach by edge-colorings of graphs},
        date={2020},
        ISSN={0166-218X},
     journal={Discrete Appl. Math.},
      volume={277},
       pages={1\ndash 13},
         url={https://doi-org.proxy.uba.uva.nl:2443/10.1016/j.dam.2019.09.004},
      review={\MR{4078905}},
}

\bib{microRNA}{article}{
      author={Chandrasekaran, Arun~Richard},
      author={Punnoose, Jibin~Abraham},
      author={Zhou, Lifeng},
      author={Dey, Paromita},
      author={Dey, Bijan~K.},
      author={Halvorsen, Ken},
       title={Dna nanotechnology approaches for microrna detection and diagnosis},
        date={2019},
     journal={Nucleic Acids Research},
      volume={47},
       pages={10489\ndash 10505},
        note={doi: 10.1093/nar/gkz580},
}

\bib{changcyles}{article}{
      author={Chang, Y.C.},
      author={Fu, H.L.},
       title={The number of 6-cycles in a graph},
        date={2003},
     journal={Bull. Inst. Comb. Appl.},
      volume={39},
       pages={27\ndash 30},
}

\bib{nanocube}{article}{
      author={Chen, J.},
      author={Seeman, N.},
       title={Synthesis from dna of a molecule with the connectivity of a cube},
        date={1991},
     journal={Nature},
      volume={350},
       pages={631\ndash 633},
}

\bib{ellis2019tile}{incollection}{
      author={Ellis-Monaghan, Joanna},
      author={Jonoska, Nata{\v{s}}a},
      author={Pangborn, Greta},
       title={Tile-based {{DNA}} nanostructures: Mathematical design and problem encoding},
        date={2019},
   booktitle={Algebraic and combinatorial computational biology},
   publisher={Elsevier},
       pages={35\ndash 60},
}

\bib{assemblydesign}{incollection}{
      author={Ellis-Monaghan, Joanna},
      author={Jonoska, Nataša},
      author={Pangborn, Greta},
       title={Tile-based dna nanostructures: Mathematical design and problem encoding},
        date={2018},
   booktitle={Algebraic and combinatorial computational biology},
   publisher={Academic Press},
}

\bib{mintiles}{incollection}{
      author={Ellis-Monaghan, Joanna},
      author={Pangborn, Greta},
      author={Beaudin, Laura},
      author={Miller, David},
      author={Bruno, Nick},
      author={Hashimoto, Akie},
       title={Minimal tile and bond-edge types for self-assembling {DNA} graphs},
        date={2014},
   booktitle={Discrete and topological models in molecular biology},
   publisher={Springer},
       pages={241\ndash 270},
}

\bib{ferrari2018}{article}{
      author={Ferrari, Margherita},
      author={Cook, Anna},
      author={Houlihan, Alana},
      author={Roulaeu, Rebecca},
      author={Seeman, Nadrian},
      author={Pangborn, Greta},
      author={Ellis-Monaghan, Joanna},
       title={Design formalism for {DNA} self-assembly of polyhedral skeletons using rigid tiles},
        date={2018},
     journal={The Journal of Mathematical Chemistry},
      volume={56},
      number={5},
       pages={1365\ndash 1392},
}

\bib{ferrari2022non}{article}{
      author={Ferrari, Margherita~Maria},
      author={Pasotti, Anita},
      author={Traetta, Tommaso},
       title={On non-isomorphic biminimal pots realizing the cube},
        date={2023},
     journal={Bull. Inst. Comb. Appl.},
      volume={98},
       pages={122\ndash 139},
}

\bib{hararycycles}{article}{
      author={Harary, F.},
      author={Manvel, B.},
       title={On the number of cycles in a graph},
        date={1971},
     journal={Mat. Casopis Sloven. Akad. Vied},
      volume={21},
       pages={55\ndash 63},
}

\bib{he2008hierarchical}{article}{
      author={He, Yu},
      author={Ye, Tao},
      author={Su, Min},
      author={Zhang, Chuan},
      author={Ribbe, Alexander~E},
      author={Jiang, Wen},
      author={Mao, Chengde},
       title={Hierarchical self-assembly of {DNA} into symmetric supramolecular polyhedra},
        date={2008},
     journal={Nature},
      volume={452},
      number={7184},
       pages={198},
}

\bib{Ellis_Next40}{incollection}{
      author={J.~Ellis~Monaghan, N.~Jonoska},
       title={From molecules to mathematics},
        date={2023},
   booktitle={Jonoska, n., winfree, e. (eds) visions of dna nanotechnology at 40 for the next 40. natural computing series.},
   publisher={Springer, Singapore},
}

\bib{jonoska}{article}{
      author={Jonoska, N.},
      author={McColm, G.L.},
      author={Staninska, A.},
       title={On stoichiometry for the assembly of flexible tile dna complexes},
        date={2011},
     journal={Natural Computing},
      volume={10},
       pages={1121\ndash 1141},
}

\bib{jonoska2006spectrum}{inproceedings}{
      author={Jonoska, Nata{\v{s}}a},
      author={McColm, Gregory~L.},
      author={Staninska, Ana},
       title={Spectrum of a pot for {DNA} complexes},
organization={Springer},
        date={2006},
   booktitle={International workshop on {DNA}-based computers},
       pages={83\ndash 94},
}

\bib{JM09}{article}{
      author={Jonoska, Nata\v~sa},
      author={McColm, Gregory~L.},
       title={Complexity classes for self-assembling flexible tiles},
        date={2009},
        ISSN={0304-3975,1879-2294},
     journal={Theoret. Comput. Sci.},
      volume={410},
      number={4-5},
       pages={332\ndash 346},
         url={https://doi.org/10.1016/j.tcs.2008.09.054},
      review={\MR{2493983}},
}

\bib{cancer}{article}{
      author={Lee, Hyojin},
      author={Dam, Duncan Hieu~M.},
      author={Ha, Ji~Won},
      author={Yue, Jun},
      author={},
      author={Odom, Teri~W.},
       title={Enhanced her2 degradation in breast cancer cells by lysosome-targeting gold nanoconstructs},
        date={2015},
     journal={ACS Nano},
      volume={9},
       pages={9859\ndash 9867},
        note={doi:10.1021/acsnano.5b05138.},
}

\bib{wheelgraphs}{misc}{
      author={Lopez, Gabriel},
      author={Johnson, Cory},
       title={Self-assembling dna complexes with a wheel graph structure},
        date={2024},
        note={https://doi.org/10.48550/arXiv.2302.13014},
}

\bib{applications}{article}{
      author={Ma, W.},
      author={Zhan, Y.},
      author={Zhang, Y.},
      author={Mao, C.},
      author={Xie, X.},
      author={Lin, Y.},
       title={The biological applications of dna nanomaterials: current challenges and future directions},
        date={2021},
     journal={Signal Transduction and Targeted Therapy},
      volume={6},
      number={351},
}

\bib{stemcell}{article}{
      author={Ma, Wnjuan},
      author={Xie, Xueping},
      author={Shao, Xiaoru},
      author={Yuxin~Zhang, Chenchen~Mao},
      author={Zhan, Yuxi},
      author={Zhao, Dan},
      author={Liu, Mengting},
      author={Li, Qianshun},
      author={Lin, Yunfeng},
       title={Tetrahedral dna nanostructures facilitate neural stem cell migration via activating rhoa/rock2 signalling pathway},
        date={2018},
     journal={Cell Proliferation},
        note={DOI: 10.1111/cpr.12503},
}

\bib{geargraphs}{article}{
      author={Mattamira, Chiara},
       title={Dna self-assembly design for gear graphs},
        date={2020},
     journal={Rose-Hulman Undergraduate Mathematics Journal},
      volume={21},
        note={https://scholar.rose-hulman.edu/rhumj/vol21/iss1/11},
}

\bib{rothemund2006folding}{article}{
      author={Rothemund, Paul~W.K.},
       title={Folding {DNA} to create nanoscale shapes and patterns},
        date={2006},
     journal={Nature},
      volume={440},
      number={7082},
       pages={297\ndash 302},
}

\bib{SK94}{article}{
      author={Seeman, Nadrian~C.},
      author={Kallenbach, Neville~R.},
       title={{DNA} branched junctions},
        date={1994},
     journal={Annual Review of Biophysics and Biomolecular Structure},
      volume={23},
      number={1},
       pages={53\ndash 86},
      eprint={https://doi.org/10.1146/annurev.bb.23.060194.000413},
         url={https://doi.org/10.1146/annurev.bb.23.060194.000413},
        note={PMID: 7919792},
}

\bib{cyclecode}{misc}{
      author={Sorrells, Jessica},
       title={Dna graph cycle counter [software]},
   publisher={GitHub},
        date={2025},
        note={https://github.com/JLSorrells88/DNAgraph-cyclecount},
}

\bib{oxLDL}{article}{
      author={Vindigni, G.},
      author={Raniolo, S.},
      author={Ottaviani, A.},
      author={Falconi, M.},
      author={Franch, O.},
      author={Knudsen, B.~R.},
      author={Desideri, A.},
      author={Biocca, S.},
       title={Receptor-mediated entry of pristine octahedral dna nanocages in mammalian cells},
        date={2016},
     journal={ACS Nano},
}

\bib{probes}{article}{
      author={wang, Dong-Xia},
      author={Wang, Jing},
      author={Wang, Ya-Xin},
      author={Du, Yi-Chen},
      author={Huang, Yan},
      author={Tang, An-Na},
      author={Cui, Yun-Xi},
      author={Kong, De-Ming},
       title={Dna nanostructure-based nucleic acid probes: construction and biological applications},
        date={2021},
     journal={Chemical Science},
      volume={12},
       pages={7602\ndash 7622},
        note={doi: 10.1039/d1sc00587a},
}

\bib{zhang1994construction}{article}{
      author={Zhang, Yuwen},
      author={Seeman, Nadrian~C.},
       title={Construction of a {DNA}-truncated octahedron},
        date={1994},
     journal={Journal of the American Chemical Society},
      volume={116},
      number={5},
       pages={1661\ndash 1669},
}

\end{biblist}
\end{bibdiv}

\end{document}